\documentclass[final,3p,times]{elsarticle}
\usepackage{graphicx}
\usepackage{amssymb}
\usepackage[utf8]{inputenc}
\usepackage[toc,page]{appendix}
\usepackage{amsthm}
\usepackage{amsmath}
\usepackage{graphicx}
\usepackage{grffile}
\usepackage{hyperref,url,color}
\usepackage{xcolor}
\usepackage[linesnumbered,algoruled,boxed,lined]{algorithm2e}
\usepackage{tikz}
\usetikzlibrary{arrows,automata,backgrounds,calendar,chains,matrix,mindmap,patterns,petri,shadows,shapes.geometric,shapes.misc,spy,trees,calc,intersections}

\newcommand{\Om}{\Omega}
\newcommand{\dw}{\,\delta\omega}

\newcommand{\dn}{\partial\n}
\newcommand{\udim}{\textbf{u}^{\star}} 
\newcommand{\pdim}{p^{\star}} 
\newcommand{\nabladim}{\nabla^{\star}\,} 
\renewcommand{\u}{\mathbf{u}} 
\renewcommand{\v}{\mathbf{v}}  
\newcommand{\uh}{\mathbf{u}_{h}} 
\newcommand{\vh}{\mathbf{v}_{h}}  

\newcommand{\xb}{x^{\star}}
\newcommand{\yb}{y^{\star}}
\newcommand{\zb}{z^{\star}}
\newcommand{\Temp}{\theta}
\newcommand{\Temph}{\theta_{h}}
\newcommand{\Tempd}{\Temp^{\star}}
\newcommand{\TempdC}{\Tempd_{C}}
\newcommand{\TempdH}{\Tempd_{H}}

\newcommand{\phidim}{\phi^{\star}}
\newcommand{\phih}{\phi_{h}}
\newcommand{\psih}{\psi_{h}}
\newcommand{\varphih}{\varphi_{h}}

\newcommand{\Cbf}{\phidim_\text{bf}}
\newcommand{\Cnp}{\phidim_\text{np}}

\newcommand{\Pe}{\text{Pe}}
\renewcommand{\Pr}{\text{Pr}}

\newcommand{\Ra}{\text{Ra}}

\newcommand{\Le}{\text{Le}}
\newcommand{\Sc}{\text{Sc}}
\newcommand{\St}{\text{St}}

\newcommand{\g}{{\mathbf g}}

\newcommand{\N}{\mathbf{N}}
\newcommand{\U}{\mathbf{U}}
\newcommand{\T}{\mathbf{T}}
\renewcommand{\P}{\mathbf{P}}

\newcommand{\rhonf}{\rho_\text{nf}}
\newcommand{\rhobf}{\rho_\text{bf}}
\newcommand{\rhonp}{\rho_\text{np}}

\newcommand{\cnp}{c_\text{np}}
\newcommand{\cbf}{c_\text{bf}}
\newcommand{\cnf}{c_\text{nf}}

\newcommand{\munfdim}{\mu^{\star}_\text{nf}}
\newcommand{\munf}{\mu_\text{nf}}
\newcommand{\mubf}{\mu_\text{bf}}

\newcommand{\knfdim}{k^{\star}_\text{nf}}
\newcommand{\knf}{k_\text{nf}}
\newcommand{\kbf}{k_\text{bf}}

\newcommand{\bnf}{\beta_\text{nf}}
\newcommand{\bnp}{\beta_\text{np}}
\newcommand{\bbf}{\beta_\text{bf}}

\newcommand{\n}{{\bf n}}
\newcommand{\aeps}{a^\varepsilon}

\newcommand{\M}{\mathcal{M}}

\newtheorem{theorem}{Theorem}
\newtheorem{proposition}{Proposition}
\journal{Journal}
\date{January 01,2021}
\begin{document}

\begin{frontmatter}

\title{Combined Newton-Raphson and Streamlines-Upwind Petrov-Galerkin iterations for nano-particles transport in buoyancy driven flow}



% \author{John Smith}\address{California, United States}
\author[labelMKR,Nucl]{M. K. RIAHI\corref{cor1}}
\ead{mohamed.riahi@ku.ac.ae}\cortext[cor1]{Corresponding author.}
\author[labelYA,Nucl]{M. Ali}
\author[labelYA,Nucl]{Y. Addad}
\author[labelEAN]{E. Abu-Nada}
%\author[labelMAS1]{M. A. Sheremet}
%\author[labelIP]{Ioan Pop}
\address[labelMKR]{Department of Applied Mathematics, Khalifa University, PO Box 127788, Abu Dhabi, UAE}
\address[labelEAN]{Department of Mechanical Engineering, Khalifa University, PO Box 127788, Abu Dhabi, UAE}
\address[labelYA]{Department of Nuclear Engineering, Khalifa University, PO Box 127788, Abu Dhabi, UAE}
\address[Nucl]{Emirates Nuclear Technology Center, Khalifa University, PO Box 127788, Abu Dhabi, UAE}
% \address[labelMAS1]{Laboratory on Convective Heat and Mass Transfer, Tomsk State University, 634050 Tomsk, Russia}
% \address[labelIP]{Department of Mathematics Babe \mathaccentV{s}-Bolyai University, Cluj-Napoca, Romania}

\begin{abstract}
The present study deals with the finite element discretization of nanofluid convective transport in an enclosure with variable properties. We study the Buongiorno model, which couples the Navier-Stokes equations for the base fluid, an advective-diffusion equation for the heat transfer, and an advection dominated nanoparticle fraction concentration subject to thermophoresis and Brownian motion forces. We develop an iterative numerical scheme that combines Newton’s method (dedicated to the resolution of the momentum and energy equations) with the transport equation that governs the nanoparticles concentration in the enclosure. We show that Stream Upwind Petrov-Galerkin regularization approach is required to solve properly the ill-posed Buongiorno transport model being tackled as a variational problem under mean value constraint. Non-trivial numerical computations are reported to show the effectiveness of our proposed numerical approach in its ability to provide reasonably good agreement with the experimental results available in the literature. The numerical experiments demonstrate that by accounting for only the thermophoresis and Brownian motion forces in the concentration transport equation, the model is not able to reproduce the heat transfer impairment due to the presence of suspended nanoparticles in the base fluid. It reveals, however, the significant role that these two terms play in the vicinity of the hot and cold walls.     
\end{abstract}


\end{frontmatter}

%\linenumbers

\section{Introduction}
\label{S:Introduction}

Natural convection, or natural circulation, is a phenomenon in which fluid recirculates due to applied temperature difference, where hot fluid (light) tends to rise up, while colder fluid (heavy) tends to fall down. This phenomenon occurs in several engineering applications such as chips cooling, large compartment ventilation, passive cooling in heat exchangers, and further ocean dynamics and weather applications. 

The suspension of nano-sized particles in base fluids, a mixture referred to as nanofluid, represents an attractive method in heat transfer engineering problems during the last two decades~\cite{minkowycz2012nanoparticle,manca2010heat,kleinstreuer2016mathematical}. Several engineering applications in heat transfer were investigated including natural convection, combined convection, heat transfer in electronic cooling, and renewable energy \cite{
kleinstreuer2013microfluidics,buongiorno2006convective,sheremet2018natural,mahian2013review,li2008thermal,bairi2018effects,li2018effects,xu2014computational,bairi2018thermal,jabbari2017thermal}. Due to the high number of publications in using nanofluids to enhance the heat transfer rate in thermal engineering systems, several reviews were conducted, such as Jabbari et al. \cite{jabbari2017thermal} and Khodadadi et al. \cite{khodadadi2018comprehensive}, Fan and Wang \cite{fan2011review}, Kakaç and Pramuanjaroenkij \cite{kakacc2009review}, Buongiorno et al. \cite{buongiorno2009benchmark}, Sheikholeslami and Ganji \cite{sheikholeslami2016nanofluid}, and Manca et al. \cite{manca2010heat}.
The role of nanofluids in augmenting the heat transfer rate in forced convection is accepted in the research community where the nanofluids are found to be very useful in enhancing the performance of forced convective flows. Conversely, the role of nanofluids is still controversial in natural convection. For example, most theoretical studies reported enhancement in heat transfer due to the dispersion of the nanoparticles in base fluids. However, several experiments reported a deterioration in heat transfer due to the addition of nanoparticles to base fluids, Wen and Ding \cite{wen2004experimental}, Li and Peterson \cite{li2010experimental}, Ho et al. \cite{ho_liu_chang_lin_2010}, and Putra et al. \cite{putra2003natural}. In some experimental measurements, a significant increase in the effective thermal conductivity for the mixture of certain types of small solid particles with a diluted water has been noticed (see \cite{ho_liu_chang_lin_2010} and \cite{chon2005empirical}). Due to this claimed significant improvement of transport properties of the nanofluids, nanofluid technology has attracted many applications including the cooling of electronic chips, cooling of smaller internal combustion engines, and biomedical technology.

Furthermore, along its development and use, the nanofluid technology has been controversial on some of its aspects. In fact, increasing the nanoparticle concentration does not necessarily mean an improvement of the heat transfer rate. This is indeed, a debate that until today researchers still does not converge to a common conclusion on the the heat transfer enhancement using nanoparticles. Nonetheless, the experimental results of Ho et al. ~\cite{ho_liu_chang_lin_2010} provide a clear evidence that heat transfer enhancement using nanofluid is actually possible with a low nanoparticles concentration (i.e. for volume fractions below 2\%). For high Rayleigh number in particular, the dispersed nanoparticles in the base fluid are able to contribute to heat transfer rate enhancement of around 18\% when compared to pure water. For concentrations equal to or greater than 2\%, on the other hand, it was observed that the addition of nanoparticles would have a negative impact on the heat transfer rate. 

One of the mostly used mathematical model used to study the heat transfer of nanofluid is the the Buongiorno model \cite{buongiorno2006convective}. In this model, the flow around the nanoparticles is regarded as a continuum. It assumes that the only mechanisms causing the nanoparticles to develop a slip velocity with respect to the base fluid are; i) the Brownian diffusion resulting from continuous collisions between the nanoparticles and the molecules of the base fluid and ii) the thermophoresis representing the nanoparticles diffusion due to temperature gradient in the domain. Other slip mechanisms considered by Buongiorno in his analysis, namely, the diffusiophoresis, the lift force, the fluid drainage, and gravity settling are considered negligible. As a result, the model translates the fact that the Brownian diffusion and thermophoresis, are the only forces responsible for the diffusivity of nanoparticle concentrations, hence, together with the main stream they update the density of the nanofluid via a convection dominated partial differential equations (PDE).

%In this paper, we focus on the numerical simulation of the heat transfer using nanofluids  following Buongiorno model. 

In this paper, we use the model Buongiorno model to investigate the heat transfer enhancement using nanofluid. We develop and analyze a new numerical iterative scheme based on a finite element method (FEM) of the steady state PDEs. The nanofluid model at hand, consists then of four equations i) the continuity equation ii) the momentum equation, which is the classical Navier-Stokes equation subject to the buoyancy force, iii) the energy equation which is the convected heat equation and iv) the nanoparticle transport equation. The latter, is an advection dominated PDEs, as per the P{\'e}clet number (ratio of the convection rate and the diffusion rate) is very high, which leads to numerical spurious while using finite difference method or FEM \cite{galeao1988consistent}. The convected dominated nanoparticle transport equation is, indeed, ill-posed and has to be regularized in order to stabilize the calculation, see for instance \cite{YURUN199747} and also \cite{ERATH2019308,TENEIKELDER20181135} for an adaptive procedure based on a posteriori error estimate.   

The momentum and energy equations are coupled through the convection and buoyancy terms, while the nanoparticle transport equation plays the role of fluid density regulator that has a crucial role in the variation of the fluid viscosity, hence on the shear stresses of the fluid. Particles migration is also impacted by the thermophoresis forces and other subs-scale forces modeled through a Brownian motion. A theoretical mathematical functional analysis of the Buongiorno model is studied in~\cite{BANSCH2020124151} where a mollified regularizing problem is set up and is shown to be weakly convergent toward the nanofluid Buongiorno model. A similar nanofluid model has been derived and studied in \cite{bansch2019thermodynamically}.  

The resulting coupled system represents the dynamics of the nanofluid inside a differentially heated cavity. We supplement the problem at hand with the following appropriate boundary conditions; non-slip velocity, adiabatic horizontal walls, specifically heated and cooled walls (Dirichlet non homogeneous conditions), and the Neumann homogeneous boundary condition for the nanoparticle transport equation describing simply the fact that none of the particles is allowed to exit the enclosure (particles-flux is null).

For the numerical discretization of such Oberbeck-Boussinesq equations we refer to~\cite{shekar2015finite,balla2016finite,ullah2020finite}. We also refer to the book \cite{girault2012finite} for a complete finite element analysis for the Navier-Stokes equations. In the present work we consider the pair {\bf P}2–{\bf P}1 continuous Taylor–Hood elements \cite{taylor1973numerical} for the discretization of the momentum equation, and consider the {\bf P}1-continuous finite elements for the energy equation. The lowest order Taylor-Hood {\bf P}2-{\bf P}1 is one of the most popular FE pairs for incompressible fluid problems. Indeed, it satisfies the inf-sup stability condition for almost regular meshes and isotropic meshes with moderate aspect ratio \cite{apel2003stability}. 

The momentum and energy equations are solved iteratively through Newton-Raphson method see for instance \cite{girault2012finite}. We use P2-continuous FE for the nanoparticle transport equation for regularization aspects that we will discuss subsequently.

The rest of the paper is organized as follows: We present in Section \ref{EqSet} the mathematical equations that govern the heat transfer in a cavity with variable properties. In section \ref{NumAlgo}, we develop the variational formulation of the governing equations. We discuss in section \ref{FEM} the finite element discretization and the overall algorithm coupling Newton and SUPG methods. Then, in section \ref{NumRes} we report the stability and convergence of our numerical scheme in addition to its validation using available experimental and numerical data. Finally we close our paper with some concluding remarks.

Throughout the paper we shall use the following notations 
\section*{Nomenclature table}
\begin{centering}
\begin{tabular}{ll}
   $(x,y)$ & coordinate system (m) \\
   $\n$ & outward normal vector \\
   $L$ & width of the cavity \\
   $\g$ & gravitational acceleration ($m s^{-2}$)\\
   $\udim$ & dimensional velocity ($ms^{-1}$)\\
   $\u$ & horizontal velocity component  \\
   $p^\star$ & dimensional pressure term \\
   $p$ & dimensionless pressure term \\
   $\Tempd$ & temperature (K) \\
   $\Temp$ & temperature   \\
   $\phidim$ & nanoparticles' volume concentration ($\%$) \\
   $\phi$ & nanoparticles' volume concentration   \\
   $\phi_{b}$ & nanoparticles' bulk volume concentration ($\%$) \\
   $\mathcal{C}$ & denoting Correlation \\
   $c_p$ & specific heat ($J\, kg^{-1} K^{-1}$)\\
     $Nu$ & Nusselt number, $Nu=-(L/\Delta \Temp)(\partial\Temp/\partial x)_{w} $  \\
  $\Ra$  & Rayleigh number  \\
  $\Pr$ & Prantl number \\
  $\Pe$ & Peclet number\\
  $\Le$ & Lewis number \\
  $\Sc$ & Schmit number  \\
%   $\St$ & 
 \end{tabular}  
\begin{tabular}{ll}   
   $\rho$ & density ($kg\, m^{-3}$)\\
  $k$ & thermal diffusivity ($m^2s^{-1}$)\\
  $\mu$ & dynamic viscosity ($N\, kg^{-1} s^{-1}$) \\
  $\alpha$ & kinematic viscosity ($m^2s^{-1}$) \\
  $\beta$ & volumetric expansion coefficient of the fluid\\
  $D^\star_{\omega}$ & dimensional Brownian diffusion coefficient \\
  $D^\star_{\Temp}$ & dimensional thermophoretic diffusion coefficient \\
  $D_{\omega}$ & Brownian diffusion coefficient \\
  $D_{\Temp}$ & thermophoretic diffusion coefficient \\
  $\pi^m$ & variables for the momentum equation\\
  $\pi^e$ & variables for the energy equation\\
  $\pi^p$ & variables for the np transport equation\\
  $^\star$ &  dimension superscript\\
  p & particle subscript \\
  nf & nanofluid subscript\\
  bf & base fluid subscript\\
  np & nanoparticle subscript\\
  C & cold wall\\
  H & hot wall
\end{tabular}
\end{centering}

\section{Equations settings}\label{EqSet}
Following Buongiorno model for incompressible nanofluid flow and using Boussinesq approximation for density, we describe the natural convection in a deferentially heated squared cavity. Thus the domain of computation $\Om$ (Figure \ref{geometry}) is a simply connected domain with Lipschitz boundary $\partial\Om$. 
It is assumed that the fluid has reached a statistically time invariant state, a reason for which we shall study the steady state of the problem. The description of the dimensional and non-dimensional problems is reported in the sequel subsections. 

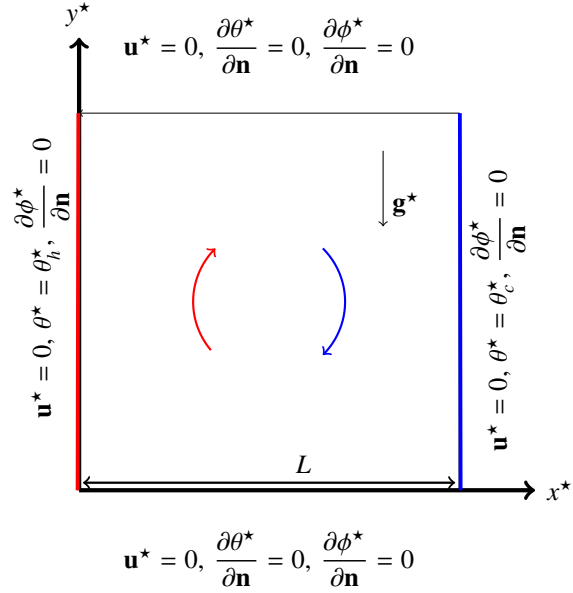
\begin{figure}[!htbp]
\centering
\begin{tikzpicture}
 \draw[->,ultra thick] (0,0)--(6,0) node[right]{$x^\star$};
 \draw[->,ultra thick] (0,0)--(0,6) node[above]{$y^\star$};
 \draw[->,ultra thin] (5,0)--(5,5)--(0,5);
 \node at (2.5,-.4) [below] {$ \udim=0,\,\dfrac{\partial\Tempd}{\partial \n}=0,\,\dfrac{\partial\phidim}{\partial\n}=0$};
  \node at (2.5,5.4) [above] {$ \udim=0,\,\dfrac{\partial\Tempd}{\partial \n}=0,\,\dfrac{\partial\phidim}{\partial\n}=0$};
  \draw[line width=0.5mm, red] (-.02,0)--(-0.01,5);
  \node at (-0.5,4.8) [left, rotate=90] {$ \udim=0,\,\Tempd=\Tempd_h,\,\dfrac{\partial\phidim}{\partial\n}=0$};
  \draw[line width=0.5mm, blue] (5.02,0)--(5.01,5);
 \node at (5.5,0.4) [right, rotate=90] {$ \udim=0,\,\Tempd=\Tempd_c,\,\dfrac{\partial\phidim}{\partial\n}=0$};
 \draw[->] (4,4.5)--(4,3.5) node[above right] {$\g^{\star}$};
   \draw[<-,line width=0.5mm,blue,thick,domain=-45:45] plot ({2.5+cos(\x)}, {2.5+sin(\x)});
   \draw[<-,line width=0.5mm,red,thick,domain=135:220] plot ({2.5+cos(\x)}, {2.5+sin(\x)});
   \draw[line width=.3mm,<->] (0.05,0.1)--(4.95,0.1) node[left=2cm, above] {$L$};
\end{tikzpicture}
\caption{Geometric sketch of the considered differential cavity and the boundary conditions for the velocity, temperature and nanoparticle concentration.}\label{geometry}
\end{figure}

\subsection{Governing dimensional equations}
Our focus will be on the long time statistically invariant state known as steady state. The flow is therefore considered time-invariant. The heat transfer equations governing the physics consist of a coupling between i) the momentum equation associated with the continuity equation, ii) the heat equation and iii) the nanoparticle transport equation and write as follows:
\begin{equation}\label{Dimensional-model}
%\footnotesize
\left\{
\begin{array}{cllr}
\nabladim\cdot \udim &\hspace{-.1in}=\hspace{-.1in}&0 & \text{ on } \Om\\
\left(\udim\cdot\nabladim\right)\udim &\hspace{-.1in}=\hspace{-.1in}& \dfrac{-1} {\rhonf} \nabladim{p^{\star}}  + \dfrac{1}{\rhonf}\nabladim\cdot\bigg( \munfdim\left(\nabladim\udim+(\nabladim\udim)^t\right) \bigg) + \g^{\star}\dfrac{(\rho_{\infty}-\rho_\text{nf})}{\rhonf}& \hspace{-.1in}\text{ on } \Om\\
 \left(\udim\cdot\nabladim \Tempd\right) &\hspace{-.1in}=\hspace{-.1in}&\dfrac{1}{(\rho c_\text{p})_\text{nf}} \nabladim\cdot\left(\knfdim(\Tempd,\phidim)\nabladim \Tempd \right) 
								+   \left( D^{\star}_{\omega}\nabladim\phidim\cdot\nabladim \Tempd+\dfrac{ D^{\star}_{\Tempd}}{\TempdC} \nabladim \Tempd \cdot \nabladim \Tempd \right)& \hspace{-.1in}\text{ on } \Om\\
	\left( \udim\cdot\nabladim \phidim\right) &\hspace{-.1in}=\hspace{-.1in}& \nabladim\cdot\left( D^{\star}_{\omega}\nabladim \phidim  + \dfrac{D^{\star}_{\Tempd}}{\TempdC}\nabladim\Tempd  \right) & \hspace{-.1in}\text{ on } \Om\\
	\rho_{\infty}-\rho_\text{nf}&\hspace{-.1in}=\hspace{-.1in}&(\rho\beta)_\text{nf}(\Tempd-\Tempd_\text{c})&
\end{array}
\right.
\end{equation}
Supplemented by the non-slip boundary condition for the fluid velocity, Dirichlet non-homogeneous for the temperature differential in two opposite sides of the domain, Neumann homogeneous for the adiabatic boundaries, and all over Neumann homogeneous for the nanoparticle concentration. Figure~\ref{geometry} showcases the considered geometric domain for the numerical simulation.   
In Eq.\eqref{Dimensional-model} the effective density, specific heat, and thermal expansion coefficients of nanofluid are given by 
\begin{eqnarray}
\rhonf  &=&  \rho_\text{bf}(1-\phidim)  +  \rho_\text{np}\phidim\label{rhonf},\\
(\rho c_\text{p})_\text{nf}  &=& (\rho c_\text{p})_\text{bf}(1-\phidim) + (\rho c_\text{p})_\text{np}\phidim,\\
(\rho\beta)_\text{nf}  &=& (\rho\beta)_\text{bf}(1-\phidim) + (\rho\beta)_\text{np}\phidim.\label{rhobetanf}
\end{eqnarray}

The Brownian diffusion coefficient $D^{\star}_{\omega}$ is defined using the Einstein-Stokes equation in the following form 
\begin{equation}\label{BrownianDb}
D^{\star}_{\omega}(\Tempd) = \dfrac{k_\text{b}}{3\pi\mu_\text{bf}(\Tempd)d_\text{p}} \Tempd,
\end{equation}
where $k_\text{b}=1.3807\cdot 10^{-23} J/K$ stands for the Boltzmann's constant. The thermophoretic diffusion coefficient $D^{\star}_{\Tempd}$ is defined as 
\begin{equation}\label{thermphoDt}
D^{\star}_{\Temp}(\Tempd,\phidim)=\dfrac{0.26 k_\text{bf}}{2k_\text{bf}+k_\text{p}}\dfrac{\mu_\text{bf}(\Tempd)}{\rho_\text{bf}}\, \phidim.
\end{equation}
The dynamic viscosity for water is defined as follows 
\begin{equation}\label{muf}
	\mu_\text{bf}(\Tempd) = 2.414 \cdot 10^{-5}\cdot 10^{247.8/(\Tempd-140)}.
\end{equation}

For the thermal conductivity and viscosity of nanofluid many correlations have been derived ~\cite{ho2010natural,abu2010effect,chon2005empirical,khanafer_vafai_2017}. These correlations are generally nonlinear with respect to their variables. For this reasons we shall keep the correlation as a general function as such for the effective viscosity for nanofluid: 
$\frac{\munf}{\mubf} = \mathcal{C}_{\munf}(\phi,\Temp)$ and for the effective thermal conductivity $\frac{\knf}{\kbf}  = \mathcal{C}_{\knf}(\phi,\Temp)$ both as a potential correlations of $\phi$ and $\Temp$.

\subsection{Dimensionless equations}
We start the non-dimensionalization of the equations by defining the following non-dimensional variables:
\begin{equation}\label{DefNonDimensional}
\begin{array}{llllll}
\u  &=& (L\slash \alpha)\, \udim&,\qquad   x &=& \xb\slash L \\
y &=& \yb\slash L &,\qquad z &=& \zb\slash L \\
\phi &=& \phidim \slash \phi_{0}&,\qquad \Temp &=& (\Tempd -\TempdC)\slash (\TempdH-\TempdC)\\ 
\munf (\Temp,\phi) &=& \munfdim(\Tempd,\phidim)\slash \mu_\text{bf}(\TempdC)&,\qquad k_\text{nf} (\Temp,\phi) &=& \knfdim(\Tempd,\phidim)\slash k_\text{bf}\\
{D}_{\omega} &=& D^{\star}_{\omega} (\Tempd)\slash D^{\star}_{\omega} (\TempdC)&,\qquad D_{\Temp} &=& D^{\star}_{\Temp} (\Tempd,\phidim)\slash D^{\star}_{\Temp} (\TempdC,\phi_{0}),
\end{array}
\end{equation}

Hence the corresponding equation writes

 \begin{equation}\label{Dimensionless-model}
%\footnotesize
\left\{
\begin{array}{clll}
\nabla\cdot \u &=&0\\
\left(\u\cdot\nabla\right)\u &=&-\pi^\text{m}_{1}(\phi) \nabla{p}  +\pi^\text{m}_{2}(\phi) \nabla\cdot\bigg(\munf(\Temp,\phi)\left(\nabla\u+(\nabla\u)^t\right) \bigg) + \pi^\text{m}_{3}(\phi)\Temp \vspace{.08in}\\
 \left(\u\cdot\nabla \Temp\right) &=& \pi^\text{e}_{1}(\phi)\nabla\cdot\left(\knf(\Temp,\phi)\nabla \Temp\right) 
								+ \left( \pi^\text{e}_{2}(\phi) {D}_{\omega}\nabla\phi\cdot\nabla \Temp+\pi^\text{e}_{3}(\phi) D_{\Temp} \nabla \Temp\cdot \nabla \Temp\right)\vspace{.13in}\\
	\left( \u\cdot\nabla \phi\right) &=&\pi^\text{p}_{1} \nabla\cdot \left(  {D}_{\omega}\nabla \phi \right)  + \pi^\text{p}_{2}\nabla\cdot\left({D}_{\Temp}\nabla\Temp  \right), 
	\end{array}
\right.
	\end{equation}
where the variable properties are defined as follows:
\begin{equation*}	
\begin{array}{llll}
% &\Nt =& D^{\star}_{\Temp}(\TempdC,\phi_{0})(\TempdH-\TempdC)\slash(\alpha \TempdC)\\
&\pi^\text{m}_{1} (\phi)=& (1-\phi + \phi\dfrac{\rhonp}{\rhobf})^{-1}\\
&\pi^\text{m}_{2} (\phi) =&\Pr \pi^\text{m}_{1} \\
&\pi^\text{m}_{3} (\phi) =& \Pr\Ra_\text{nf} \left( \dfrac{(1-\phidim)}{ (1-\phidim)+\phidim \dfrac{\rhonp}{\rhobf}}
+ \dfrac{ \phidim}{ (1-\phidim) +\phidim \dfrac{\rhonp}{\rhobf}}\dfrac{\bnp}{\bbf} \right) \\
&\pi^\text{e}_{1} (\phi) =& (1-\phi + \phi\dfrac{\rhonp\Cnp}{\rhobf\Cbf})^{-1}\\
&\pi^\text{e}_{2} (\phi) =& \dfrac\Pr\Sc\phi_\text{b} \left((1-\phi)\dfrac{\rhobf\Cbf}{\rhonp\Cnp} +\phi \right)  \vspace{0.1in}\\
&\pi^\text{e}_{3} (\phi) =& \St \dfrac{\Pr}{\Sc} \left(\dfrac{\TempdH-\TempdC}{\TempdH}\right) 
						\left((1-\phi)\dfrac{\rhobf\Cbf}{\rhonp\Cnp} +\phi\right)^{-1} \\
&\pi^\text{p}_{1} =& \dfrac{D_{\omega^0}} {\alpha} =\dfrac{1}{\Le} \vspace{0.1in}\\
&\pi^\text{p}_{2} =& \St \dfrac{\Pr}{\Sc} \left(\dfrac{\TempdH-\TempdC}{\TempdC}\right) \dfrac{1}{\phi_\text{b}}.
\end{array}
\end{equation*}
and the non-dimensional numbers are given by: 
\begin{equation*}	
\begin{array}{llll}
&\text{Rayleigh number} \quad & \Ra =& g(\rho\beta)_\text{bf}(\TempdH-\TempdC) L^{3} \slash (\mu_\text{bf}(\TempdC)\alpha),\\
&\text{Prantle number} \quad & \Pr =& \mu_\text{bf}(\TempdC)\slash(\rho_\text{bf}\alpha),\\
&\text{Lewis number} \quad & \Le =& \alpha\slash D^{\star}_{\omega}(\TempdC),\\
&\text{Peclet number} \quad & \Pe_\text{np} =&  {\|\u\|_{2}}\slash{\left(\pi_{1}^{p}D_{\omega}\right)}.\\
\end{array}
\end{equation*}
We proceed with the numerical discretization of the governing equations. We propose to decouple the whole system into two parts; The first part consists in solving the momentum equations along with the energy equation, the second part consists in solving the stabilized nanoparticle equation with an additional constrain to assess that the bulk amount of nanoparticle in the fluid is constant during the calculation. The main motivation behind this decoupling is to reduce the non-linearity of the equations (momentum and energy) in term of the nanoparticle volume fraction and give room to the numerical scheme to stabilize itself iteratively for the transport equation. Finally, we re-iterate through these two parts in order to update the coefficients and the density of the fluid. More details are described in section \ref{NumAlgo} below.
%~~~~~~~~~~~~~~~~~~~~~~~~~~~~~~~~
\section{Variational formulations and approximation tools settings}\label{NumAlgo}
%~~~~~~~~~~~~~~~~~~~~~~~~~~~~~~~~
The numerical scheme we propose in this work is based on the finite element discretization of the dimensionless equations. As the steady problem at hand presents non-linearity in term of the velocity and temperature, a direct linear algebra solver, is, therefore, impractical to solve the coupled system.  One has to relax the non-linearity and set up an iterative numerical scheme that converges to the desired solution. Newton's method plays a key role here and has been shown to be very efficient in term of convergence \cite{girault2012finite}. 

We consider homogeneous Dirichlet boundary conditions for the velocity, i.e. $u=0$ on $\partial\Om$. 
Let us consider the following Hilbert spaces for the temperature, velocity and pressure as follows: 
$$\mathbf{T}=H^{1}(\Om) \quad, {\bf U}=H^{1}_{0}(\Om)\times H^{1}_{0}(\Om), \quad {\bf P}=\{ p \in L^{2}(\Om)\, | \int_{\Om} p \dw=0\,\}.$$
where $$L^{2}(\Om)=\{u \,|\, \int_{\Om}|u|^2 \dw < \infty \},$$ and 
$$ H^{1}(\Om)=\left\{u\in L^{2}(\Om) \,| \,\nabla u \in L^{2}(\Om)\right\}, $$ and 
$$H^{1}_{0}(\Om)=\{u\in H^{1}(\Om) \,|\, u_{|\partial\Om}=0 \}.$$
The Sobolev space $L^2(\Om)$ is endowed with the usual inner product denoted by the duality pairing $(\cdot,\cdot)$, that generates the $L^2$-norm $\|\cdot\|_{2}$.

\subsection{Newton's (optimize then discretize) Methods for the solution of the momentum and energy equations}
In order to present a general algorithm, we shall separate nanoparticles equation from the momentum and energy equations while we are processing. The generality here is mainly targeting the use of any correlations (as there is a variety available in literature). This gives our approach a flexibility to treat different type of nanofluid and consider different (possibly highly nonlinear) thermal and viscosity correlations.
The idea behind the following notation and calculation of the tangent equations is that we formulate the problem as a multivariable function $\mathbf{F}$ and use the classical Newton-Raphson iterations that reads  
$$\begin{bmatrix}\u^{k+1}\\ p^{k+1}\\ \Temp^{k+1}\end{bmatrix} = \begin{bmatrix}\u^{k}\\ p^{k}\\ \Temp^{k}\end{bmatrix} - \left(\mathbf{DF}\left((\u^{k},p^{k},\Temp^{k})\right)\right)^{-1} \mathbf{F}\left((\u^{k},p^{k},\Temp^{k})\right),$$
where $\mathbf{DF}(X^{k})$ stands for the Jacobian (isomerism in ${\bf U}\times{\bf P}\times{\bf T}$).

In practice, we define the coupled momentum-energy variational form as such for every trial $(\v,q,\zeta)\in {\bf U}\times{\bf P}\times{\bf T}$ we have
\begin{equation}
\begin{array}{lll}
 \mathcal{F}_{(\v,q,\zeta)}:{\bf U}\times{\bf P}\times{\bf T}\longrightarrow\mathbb{R}\\
(\u,p,\Temp)\longmapsto \mathcal{F}_{(\v,q,\zeta)}(\u,p,\Temp):=&  ((\u\cdot\nabla\u), \v )
 + \pi^{m}_{2} \left( \munf(\Temp,\phi) \left(\nabla\u+(\nabla\u)^t\right) , \nabla \v \right)\\
\quad&- ( \pi^{m}_{1} \nabla p, \v ) -\left(  \pi^{m}_{1} \u , \nabla q \right)
	  -\left( \pi^{m}_{3} \Temp (\u\cdot \textbf{e}_{z}),\v\right) + \varepsilon ( p , q )\\
\quad& + (  \left(\u\cdot\nabla\Temp\right),\zeta ) 
                 +\displaystyle ( \pi^\text{e}_{1}(\phi)\left(\knf(\Temp,\phi)\nabla \Temp\right) \cdot \nabla \zeta   ) \\
\quad&+ ( \pi^\text{e}_{2}(\phi) {D}_{\omega} (\nabla\phi , \nabla \Temp), \zeta )
	   +\displaystyle\left(\pi^\text{e}_{3}(\phi) D_{\Temp}  \left(\nabla\Temp\cdot \nabla \Temp\right),\zeta\right).
% \\ \quad&= 0.
\end{array}
\end{equation}
Then we define the vector field $\mathbf{F}(\u,p,\Temp)\in {\bf U}\times{\bf P}\times{\bf T}$ 
as the vector that satisfies the following inner product formula for every $(\v,q,\zeta)\in {\bf U}\times{\bf P}\times{\bf T}$
\begin{equation*}
   ( \mathbf{F}(\u,p,\Temp), 
   \begin{bmatrix}\v\\ q\\ \zeta\end{bmatrix})= \mathcal{F}_{(\v,q,\zeta)}(\u,p,\Temp).
\end{equation*}

Besides, we define the tangent coupled momentum-energy variational trilinear form as such for every trial $(\v,q,\zeta)\in {\bf U}\times{\bf P}\times{\bf T}$ we have
\begin{equation}
\begin{array}{lll}
 \mathcal{DF}_{(\v,q,\zeta)}(\u,p,\Temp):{\bf U}\times{\bf P}\times{\bf T}\longrightarrow\mathbb{R}\\
\mathcal{DF}_{(\v,q,\zeta)} (\u,p,\Temp)(\delta \u,\delta p,\delta \Temp)&:=& \displaystyle\left( \left(\delta\u\cdot\nabla\u\right) , \v \right) 
 +\displaystyle\left( \left(\u\cdot\nabla\delta\u\right) , \v \right)  \\
&&+\displaystyle\left( \pi^{m}_{2} \munf(\Temp,\phi) \left(\nabla\delta\u+(\nabla\delta\u)^t\right) , \nabla \v \right)\\
&&- \displaystyle\left( \pi^{m}_{1} \nabla \delta p ,\v \right) -\left(  \pi^{m}_{1} \delta\u,\nabla q \right)\\
&&  -\displaystyle\left( \pi^{m}_{3} \delta\Temp (\u\cdot \bold{e}_{z}) , \v\right) 
    -\displaystyle\left( \pi^{m}_{3} \Temp (\delta\u\cdot \bold{e}_{z}) , \v\right) 
+ \varepsilon\left(  \delta p , q\right)\\
&& + \displaystyle\left(  \left(\delta\u\cdot\nabla\Temp\right),\zeta \right)
+ \displaystyle\left(  \left(\u\cdot\nabla\delta\Temp\right),\zeta \right)\\
&&+\displaystyle\left( \pi^\text{e}_{1}(\phi)\left(\knf(\Temp,\phi)\nabla\delta \Temp\right) , \nabla \zeta  \right) \\
&&+\displaystyle\left( \pi^\text{e}_{2}(\phi) {D}_{\omega}\left(\nabla\phi\cdot\nabla\delta \Temp\right), \zeta \right) 
+\displaystyle\left(\pi^\text{e}_{3}(\phi) 2D_{\Temp}  \left(\nabla\Temp\cdot \nabla\delta \Temp\right),\zeta \right) ,
\end{array}
\end{equation}
and we define the linear operator $\mathbf{DF} (\u,p,\Temp)$ that satisfies the following equation
$$
(\mathbf{DF} (\u,p,\Temp)(\delta \u,\delta p,\delta \Temp),
\begin{bmatrix}\v\\ q\\ \zeta\end{bmatrix}) = \mathcal{DF}_{(\v,q,\zeta)} (\u,p,\Temp)(\delta \u,\delta p,\delta \Temp).
$$
$\forall (\v,q,\zeta)\in {\bf U}\times{\bf P}\times{\bf T}$.

 \subsection{The nanoparticles transport equation: variational formulation}
 The variational formulation for the nanoparticle concentration transport equation reads: For avery $\varphi\in X:=H^{1}(\Om)$ find $\phi\in H^{1}(\Om)$ such that  
\begin{equation}
\begin{array}{lll}
a(\phi,\varphi)&:=&\displaystyle\left(\left( \u\cdot\nabla \phi\right), \varphi \right) 
+ \displaystyle\left( \pi^\text{p}_{1} \left(  {D}_{\omega}\nabla \phi \right) , \nabla\varphi \right) \\
&&- \displaystyle\left( \pi^\text{p}_{1}{D}_{\omega} \dfrac{\partial \phi}{\partial\bf n} \right)_{\Gamma}
     - \displaystyle\left( \pi^\text{p}_{2}{D}_{\Temp} \dfrac{\partial \Temp}{\dn} \right)_{\Gamma} =0.\\
(b(\Temp),\varphi) &=& \pi^\text{p}_{2}\left( {D}_{\Temp}\nabla\Temp ,\nabla\varphi \right).
% + \displaystyle\left(\pi^\text{p}_{2}\left({D}_{\Temp}\nabla\Temp  \right) , \nabla\varphi \right)     
\end{array}
\end{equation}
After applying the boundary conditions. This problem is not well posed. Indeed, it does not satisfy the inf-sup condition. We shall discuss this concern in the sequel where we introduce a regularization parameter $\varepsilon$ acting as a reaction term. In the numerical code this parameter is taken very small ($\varepsilon\approx 1.e-10$) to ensure satisfaction of the inf-sup condition and not to harm the model, as the nanoparticles concentration does not exhibit a reaction within the mixture.   

\section{Finite element algorithm} \label{FEM}
For the steady states under consideration we use i) standard Taylor-Hood finite element approximation \cite{taylor1973numerical} for the space discretization of the Navier-Stokes equations, approximating the velocity field with \textbf{P2} finite elements, and the pressure with the P1 finite element.
%where K is an element of the triangulation Th. The other variables (temperature, enthalpy) are discretized using P1finite elements.
We assume we have the triangulation $\mathcal{T}_{h}$ of the computational domain $\Om$, such that 
$$\Om=\cup_{n=e}^{nel}\overline{\Om^{e}}$$
we seek for approximated solution over the finite dimensional vector spaces $\U_{h}\times\P_{h}\times \T_{h} \subset \U\times\P\times\T$, where $h$ denotes the discretization parameter. 
We denote by $(\uh,p_{h},\Temph)$ (respectively $(\vh,q_{h},\zeta_{h})$) the discrete FE solution approximating the continuous solution $(\u,p,\Temp)$ (respectively $(\v,q,\zeta)$). We also denote by $\phih$ the FE approximation of $\phi$.

\subsection{Matrix assembly for Newton's method}
Hereafter, we associate to the linear operator $\mathbf{DF} (\u,p,\Temp)(\cdot,\cdot ,\cdot)$ a matrix representation as follows:

\begin{equation}\label{dfsystem}
    \begin{bmatrix}
    DF_{h}^{\u,\u}   & DF_{h}^{\u,p}   & DF_{h}^{\u,\Temp}\\
    DF_{h}^{p,\u}   & DF_{h}^{p,p}   & 0 \\
    DF_{h}^{\Temp,\u}& 0& DF_{h}^{\Temp,\Temp} 
    \end{bmatrix}
    \begin{bmatrix}
    \delta\uh \\ \delta p_{h} \\ \delta \Temph
    \end{bmatrix},
\end{equation}
where the blocks in this matrix are linear operators associated to the following bilinear forms as follows   
\begin{eqnarray*}
( DF_{h}^{\u,\u} \delta\uh,\vh ) &=& \left( \left(\delta\uh ,\nabla\uh\right) , \vh \right)
 +( \left(\uh,\nabla\delta\uh\right), \vh ) + ( \pi^{m}_{2} \munf(\Temph,\phih)\left(\nabla\delta\uh+(\nabla\delta\uh)^t\right) , \nabla \vh ) \\
( DF_{h}^{\u,p} \delta p_{h},\vh ) &=&  ( \pi^{m}_{1} \nabla \delta p , \vh)\\
( DF_{h}^{\u,\Temp} \delta\Temp ,\vh )&=& -( \pi^{m}_{3} \delta\Temp (\uh , \bold{e}_{z}),\vh) \\
( DF_{h}^{p,\u} \delta\uh,q_{h} ) &=& -\left(  \pi^{m}_{1} \delta\uh,\nabla q_{h} \right)\\
( DF_{h}^{p,p} p_{h},q_{h}) &=&  \varepsilon(  \delta p_{h} , q_{h})\\
( DF_{h}^{\Temp,\u} \delta\Temph, \zeta_{h} )&=& (  (\uh,\nabla\delta\Temph) , \zeta_{h})\\
( DF_{h}^{\Temp,\Temp} \delta\Temph,\zeta_{h} )&=& 
\displaystyle( \pi^\text{e}_{1}(\phi)(\knf(\Temph,\phih)\nabla\delta \Temph) , \nabla \zeta_{h}  )\\
&&+( \pi^\text{e}_{2}(\phih) {D}_{\omega}\left(\nabla\phih,\nabla\delta \Temph\right) \zeta_{h})
        +(\pi^\text{e}_{3}(\phi) 2D_{\Temph}  \left(\nabla\Temph, \nabla\delta \Temph\right),\zeta_{h}).
\end{eqnarray*}
Newton's iterations updates the FE solution $(\uh,p_{h},\Temph)$ as follows 

\begin{equation}\label{MomentumEnergyUpdate}
    \begin{bmatrix}
    \uh^{k+1} \\ p_{h}^{k+1} \\ \Temph^{k+1}
    \end{bmatrix}
    =
        \begin{bmatrix}
    \uh^{k} \\ p_{h}^{k} \\ \Temph^{k}
    \end{bmatrix} - 
    \begin{bmatrix}\delta\u^{k} \\ \delta p^{k} \\ \delta \Temp^{k}\end{bmatrix}
\end{equation}
where $(\delta\uh^{k},\delta p_{h}^{k},\delta\Temph^{k})$ is solution to 
\begin{equation}\label{MomentumEnergydeltaUpdate}
        \begin{bmatrix}
    DF_{h}^{\u,\u}   & DF_{h}^{\u,p}   & DF_{h}^{\u,\Temp}\\
    DF_{h}^{p,\u}   & DF_{h}^{p,p}   & 0 \\
    DF_{h}^{\Temp,\u}& 0& DF_{h}^{\Temp,\Temp} 
    \end{bmatrix}
    \begin{bmatrix}\delta\uh^{k} \\ \delta p_{h}^{k} \\ \delta \Temph^{k}\end{bmatrix}
    = 
    \begin{bmatrix}
    F_{h}^{\u\u}   & F_{h}^{\u p}   & F_{h}^{\u\Temp}\\
    F_{h}^{p\u}   & F_{h}^{pp}   & 0 \\
    F_{h}^{\Temp\u}& 0& F_{h}^{\Temp\Temp}
    \end{bmatrix}
    \begin{bmatrix}\uh^{k} \\ p_{h}^{k} \\ \Temph^{k}\end{bmatrix}
\end{equation}
with right-hand side (presented as a matrix-vector product) contrains blocks of linear operators associated to the following bilinear forms as such
\begin{eqnarray*}
( F_{h}^{\u\u} \uh,\vh)&=& \left( \left(\uh\cdot\nabla\uh\right), \vh \right)
 +\displaystyle\left( \pi^{m}_{2} \munf(\Temph,\phih) \left(\nabla\uh+(\nabla\uh)^t\right) , \nabla \vh \right)\\
( F_{h}^{\u p} p_{h},\vh)e &=& \left( \pi^{m}_{1} \nabla p_{h},\vh \right)\\
( F_{h}^{\u\Temp} \Temph,\vh) &=& -\left( \pi^{m}_{3} \Temph (\uh, \bold{e}_{z}) , \vh\right)\\
( F_{h}^{p\u} \uh , q_{h}) &=& -\left(  \pi^{m}_{1} \uh,\nabla q_{h} \right)\\
( F_{h}^{pp} p_{h},q_{h}) &=& \varepsilon\left(  p_{h} , q_{h} \right)\\
( F_{h}^{\Temp\u} \uh,\zeta_{h}) &=& \left(  \left(\uh\cdot\nabla\Temph\right), \zeta_{h} \right)\\
( F_{h}^{\Temp\Temp} \Temph,\zeta_{h}) &=& \displaystyle\left( \pi^\text{e}_{1}(\phi)\left(\knf(\Temph,\phih)\nabla \Temph\right) , \nabla \zeta_{h}  \right) 
+\displaystyle\left( \pi^\text{e}_{2}(\phih) {D}_{\omega}\left(\nabla\phih,\nabla \Temph\right), \zeta_{h}\right)\\
&&	 +\displaystyle\left(\pi^\text{e}_{3}(\phih) D_{\Temph}  \left(\nabla\Temph, \nabla \Temph\right),\zeta_{h}\right).
\end{eqnarray*}

\subsection{Ill-posedness of the Buongiorno nanoparticle transport model with FEM}
We develop hereafter a numerical scheme that approximates the solution of the transport dominated advection-diffusion nanoparticles concentration equation. Despite the fact that finite element solution is very adaptive for this kind of governing equations, such method and other Galerkin based approaches may potentially fail where their relative discrete solution will be market by spurious oscillations in space. This happens if their element P{\'e}clet number goes beyond certain critical value. Other discretization methods suffer, actually, from this issue, for instance the finite difference technique. In the later, the spatial oscillation could be reduced by upwinding scheme. For the finite element method, we have equivalent methods to the later upwinding scheme, such as Petrov-Galerkin and Streamline-Upwind Petrov-Galerkin (SUPG)~\cite{franca2004stabilized,JOHN2013289,ERATH2019308}. Here, the shape function is modified in order to mimic the upwinding effect and therefore reduce, or eliminate, the spatial oscillations. An other interesting stabilization method is the Galerkin/Least-Square (GLS) see for instance \cite{TENEIKELDER20181135,TENEIKELDER2018259} and references therein. All the above and other techniques \cite{BURMAN20101114,BOCHEV20042301,RUSSO20061608} more adaptable for the time-dependent equations, use artificial viscosity in the direction of the streamlines.

Let us define the space 
$$X_{h}=\left\{ \varphih\in \mathcal{C}^0(\overline{\Om}); \,\forall \Om^{e}\in\mathcal{T}_{h},\, \varphi_{h|\Om^{}e}\in \mathbf{P}_2\right\}\subset X,$$
Over which, the bilinear form related to the variational formulation of the advection-diffusion nanoparticle concentration writes as follows: for any trial function $\varphih\in X_{h}$ find $\phih\in X_{h}$
\begin{equation}\label{a_bilinearform}
    \aeps(\phih,\varphih) = b(\Temph)
\end{equation}
where 
\begin{eqnarray*}
\aeps(\phih,\varphih)&=& \left(\left( \u\cdot\nabla \phih^{}\right) , \varphih \right) 
+ \pi^\text{p}_{1}\left(  {D}_{\omega}^{}\nabla \phih^{} , \nabla\varphih \right) + \varepsilon\left(  \phih , \varphih \right).\\
(b(\Temph),\varphih)  &=& \pi^\text{p}_{2}\left( {D}_{\Temph}\nabla\Temph ,\nabla\varphih \right)
\end{eqnarray*}
We endow $X$ with the following norm 
  $$\|\varphi\|_{\varepsilon}:= \sqrt{\pi_{1}^{p} \|\nabla\varphi\|_{2}^2+\varepsilon \|\varphi\|_{2}^{2}},$$
 which we will use in the sequel. 
\begin{proposition}\label{propositionWellPosedness}(Well-posedness)
The bilinear form $\aeps(\cdot,\cdot)$  is coercive such that 
\begin{eqnarray*}
\aeps(\varphih,\varphih) &\geq&  \|\varphih\|_{\varepsilon}^{2},\\
\aeps(\varphih,\varphih) &\geq& \alpha_{p} \|\varphih\|_{2}^{2}.
\end{eqnarray*}
\end{proposition}
\begin{proof}
We have 
$$\aeps(\varphih,\varphih)= \left(\left( \uh\cdot\nabla \varphih \right),\varphih \right) 
+ \pi^\text{p}_{1}  \left( {D}_{\omega}^{}\nabla \varphih ,\nabla\varphih \right) + \varepsilon\left(  \varphih,\varphih \right)$$
Note that we have $\left(\left( \uh\cdot\nabla \varphih\right) \varphih \right)=\left(\nabla\cdot\left(\uh \varphih^2\right),\frac 1 2\right)=-\frac{1}{2}\left((\uh\cdot\n)\varphih^{2}\right)_{\Gamma}=0$ after using integration by part and the divergence theorem together with the velocity homogeneous boundary condition. We therefore have 
\begin{eqnarray*}
\aeps(\varphih,\varphih) 
&\geq&   \pi_{1}^{p} \|\nabla\varphih\|^{2}_{2}+ \varepsilon \|\varphih\|_{2}^{2} -\frac{1}{2} \left((\uh\cdot\n)\varphih^{2}\right)_{\Gamma} \\
&\geq&  \dfrac{\pi_{1}^{p} +\varepsilon c_{\Om}}{c_{\Om}} \|\varphih\|^{2}_{2}.
\end{eqnarray*}
Where $c_{\Om}$ stands for the Poincar\'e constant. Finally, by setting $\alpha_p=\frac{\pi_{1}^{p} +\varepsilon c_{\Om}}{c_{\Om}}$, we obtain the coercivity result in $X$ and in $L^2(\Omega)$ as stated above. 
\end{proof}
In the subsequent part, we shall discuss the stabilization of the transport equation for the concentration of nanoparticles. Indeed, we use the so-called Streamline Upwind Petrov-Galerkin (SUPG) method that enlarges the standard trial space via streamline trials. In practice, the SUPG formulates as follows. 
Given $\uh\in \U_h$, for a test function $\varphi$ in 
$$\text{span}\left\{\varphih + \sum_{\Om^{e}} \delta_{\Om^{e}} \uh\nabla\varphih\right\}$$
find $\phih$ such that 
\begin{equation}\label{varprbsupg}
\aeps_{\sc supg}(\phih,\varphih)= b_{\sc supg}(\phih,\varphih).
\end{equation}
where 
\begin{eqnarray*}
 \aeps_{\sc supg}(\phih,\varphih)&:=&  \aeps(\phih,\varphih) + \sum_{e}^{nel}\int_{\Om^{e}}\delta_{\Om^{e}}\left( \uh\cdot\nabla \phih +\varepsilon\phih^2 - \pi^\text{p}_{1} \nabla\cdot\left({D}_{\omega}^{}\nabla \phih\right)\right) \left(\u\nabla\varphih\right) \dw 
% &&\hspace{.5cm}- \int_{\Om^{e}} \delta_{\Om^{e}}\pi^\text{p}_{1} \nabla\cdot\left({D}_{\omega}^{}\nabla \phi\right) \left(\u\nabla\varphi\right) \dw\\
% &&\hspace{.5cm}+\int_{\Om^{e}}\delta_{\Om^{e}} \varepsilon\varphi^2 \dw.
\end{eqnarray*}
and 
\begin{equation*}
    b_{\sc supg}(\phih,\varphi)=b(\phih,\varphi)+\sum_{e}^{nel}\int_{\Om^{e}}\delta_{\Om^{e}}\pi^\text{p}_{2} \nabla\cdot\left({D}_{\Temph}\nabla\Temph\right) \left(\u\nabla\varphih\right) \dw.
\end{equation*}
Here, the integral contributions stand for the SUPG formulation \cite{brooks1982streamline,franca1992stabilized,gelhard2005stabilized,burman2011analysis,burman2010consistent}, where, $\delta_{\Om^{e}}$ are user-chosen weights element dependent parameters. The above convection dominated nanoparticle equation, even if it is linear with respect to the unknown $\phih$ sill exhibit spurious numerical oscillations due to the fact that the dominant contribution comes from the convective term, which is far away larger than the thermophoresis effect and the diffusive Brownian motion. This directly reflects that the corresponding P{\'e}clet number $\Pe$ (non-dimensional quantity), i.e., the ratio of the convection rate and the diffusion rate is relatively high. To overcome this handicap in the numerical simulation, we use the SUPG method, which provides numerical stability. Indeed, the SUPG adds an artificial stream-diffusion along the streamlines of the particles flow. In practice the SUPG weight function parameter is adjusted in term of local P{\'e}clet number. In theory, the choice of the weighted parameter function must satisfy 
\begin{equation}\label{ubweightedfunction}
\delta_{\Om^{e}} \leq \min\left\{ \dfrac{1}{2\varepsilon} , \dfrac{h_{\Om^{e}}^2\| D_{\omega}\|_{\infty,\Om^{e}}^{2}}{2 C_{inv}^2} \right\}
\end{equation}
a condition for which we ensure the coercivity and hence the well-posedness of the SUPG problem via the next coercivity result in the Banach space $V_{h}$:
$$V_{h}=\left\{ \varphih\in X_{h}(\Om) \, | \, \u\cdot\nabla{\varphih} \in L^{2}(\Om) \right\}$$
endowed with the SUPG-norm defined by: 
\begin{equation}\label{supgnorm}
    \|\varphih\|_{\sc supg} :=\sqrt{\pi^\text{p}_{1}\|\nabla \varphih\|_{2}^{2} + \varepsilon\|\varphih\|_{2}^{2} + 
\|\delta_{\Om}^{\frac{1}{2}} \u\cdot\nabla \varphih\|_{2}^{2}}
\end{equation}

As our numerical approach is based on Newton's method, which increasingly varies the Rayleigh number $\Ra$, our SUPG weighted function is made as function of the local P{\'e}clet number $\Pe_{\text{np},\Om^{e}}$, the nanofluid Rayleigh number $\Ra_\text{nf}$ and the local element size $h_{\Om^{e}}$. Hence the numerically chosen weighted function is  
$$\delta_{\Om^{e}}=h_{\Om^{e}}^2\dfrac{\Ra_\text{nf}}{\Pe_{\text{np},\Om^{e}}}.$$

% By considering $C_{inv}$ a moderated constant (of equivalency of norms in finite dimensional vector space) such as 
% $$C_{inv} = \sqrt{\dfrac{Pr_\text{np}\alpha^2}{\g}}$$
% we have the following upper bound estimate of the numerically chosen weighted function.
% $$\dfrac{\delta_{\Om^{e}}}{h^2_{\Om^{e}}} 
% =\dfrac{\Ra_\text{nf}}{\Pe_{\text{np},\Om^{e}}}
% % \leq \dfrac{10}{\Pr_\text{np} \Le^2}
% % \leq \dfrac{1}{\Pr_\text{np} \Le^2}\dfrac{\Pr_\text{np}\alpha^2}{C_{inv}^2} 
% \leq  \dfrac{\| D_{\omega}\|_{\infty,\Om^{e}}^{2}}{2 C_{inv}^2} 
% % \approx \dfrac{\alpha^2}{\Le^2 C_{inv}^2} .
% $$

\begin{theorem}
The bilinear form  $\aeps_{\sc supg}(\cdot,\cdot)$ satisfies the following Lax-Milgram conditions 
\begin{eqnarray}
 \hbox{Coercivity} \quad  & \aeps_{\sc supg}(\varphih,\varphih) \geq \dfrac{1}{2}\|\varphih\|^{2}_{\sc supg}\\
 \hbox{Continuity} \quad & \aeps_{\sc supg}(\varphih,\psi) \leq C\|\varphih\|_{\sc supg} \|\psi\|_{\sc supg}
\end{eqnarray}
for which the problem \eqref{varprbsupg} is well-posed and has unique solution.
\end{theorem}    
\begin{proof}
The proof follows the classical bounds estimates for SUPG analysis.
\begin{eqnarray}
\aeps_{\sc supg}(\varphih,\varphih) &\geq&  
 \left(\left( \u\cdot\nabla \varphih\right), \varphih \right) 
+ \pi^\text{p}_{1}\left( {D}_{\omega}\nabla \varphih , \nabla\varphih \right) + \varepsilon\left(\varphih,\varphih\right) \notag
\\&&+\sum_{e}^{nel} 
\int_{\Om^{e}}\delta_{\Om^{e}}\left( \u\cdot\nabla \varphih\right) \left(\u\nabla\varphih\right) \dw \notag\\
&& - \sum_{e}^{nel} \left|\int_{\Om^{e}} \left(\varepsilon\varphih^2 - \pi^\text{p}_{1} \nabla\cdot\left({D}_{\omega}\nabla \varphih\right)\right) \left(\delta_{\Om^{e}}^{\frac{1}{2}}\u\nabla\varphih\right) \dw \right|\\
 &\geq&  
  \pi^\text{p}_{1}\|\nabla \varphih\|_{2}^{2} + \varepsilon\|\varphih\|_{2}^{2} + 
\|\delta_{\Om}^{\frac{1}{2}} \u\cdot\nabla \varphih\|_{2}^{2} \notag\\
&& - \sum_{e}^{nel} \left|\int_{\Om^{e}} \left(\varepsilon\varphih^2 - \pi^\text{p}_{1} \nabla\cdot\left({D}_{\omega}\nabla \varphih\right)\right) \left(\delta_{\Om^{e}}^{\frac{1}{2}}\u\nabla\varphih\right) \dw \right|\\
 &\geq&  
  \|\varphih\|_{\sc supg}^{2} - \sum_{e}^{nel} \left|\int_{\Om^{e}} \left(\varepsilon\varphih^2 - \pi^\text{p}_{1} \nabla\cdot\left({D}_{\omega}\nabla \varphih\right)\right) \left(\delta_{\Om^{e}}^{\frac{1}{2}}\u\nabla\varphih\right) \dw \right|
\label{lowerbound2}
\end{eqnarray}
where we have use the fact that $D_{\omega}>1$ almost every where as per its definition \eqref{DefNonDimensional} and \eqref{BrownianDb} together with the maximum principle for the (heat) energy equation.

On the other hand, we have the following inequality 
\begin{eqnarray}
&&\left|\int_{\Om^{e}} \delta_{\Om^{e}}^{\frac{1}{2}}\left(\varepsilon\varphih^2 - \pi^\text{p}_{1} \nabla\cdot\left({D}_{\omega}\nabla \varphih\right)\right) \left(\delta_{\Om^{e}}^{\frac{1}{2}}\u\nabla\varphih\right) \dw \right|\notag\\
&\leq& \int_{\Om^{e}}\left| \delta_{\Om^{e}}^{\frac{1}{2}}\left(\varepsilon\varphih^2 - \pi^\text{p}_{1} \nabla\cdot\left({D}_{\omega}\nabla \varphih\right)\right) \left(\delta_{\Om^{e}}^{\frac{1}{2}}\u\nabla\varphih\right) \right| \dw \notag\\
&\leq& \left(
\varepsilon\delta_{\Om^{e}}^{\frac{1}{2}} \|\varphih\|_{2,\Om^{e}} + \pi^\text{p}_{1}\delta_{\Om^{e}}^{\frac{1}{2}}\|D_{\omega}\|_{\infty,\Om^{e}}\| \Delta\varphih\|_{2,\Om^{e}} \right) \|\delta_{\Om^{e}}^{\frac{1}{2}}\u\nabla\varphih\|_{2,\Om^{e}}\label{hucs1}\\
&\leq& \left(
\varepsilon\delta_{\Om^{e}}^{\frac{1}{2}} \|\varphih\|_{2,\Om^{e}} + \delta_{\Om^{e}}^{\frac{1}{2}} \dfrac{\pi^\text{p}_{1}\|D_{\omega}\|_{\infty,\Om^{e}} C_{inv}}{h_{\Om^{e}}} \| \nabla\varphih\|_{2,\Om^{e}} \right) \|\delta_{\Om^{e}}^{\frac{1}{2}}\u\nabla\varphih\|_{2,\Om^{e}}\label{young}\\
&\leq& \left(
\sqrt{\dfrac{\varepsilon}{2}} \|\varphih\|_{2,\Om^{e}} +   \sqrt{\dfrac{\pi^\text{p}_{1}}{2}} 
\| \nabla\varphih\|_{2,\Om^{e}} \right) \|\delta_{\Om^{e}}^{\frac{1}{2}}\u\nabla\varphih\|_{2,\Om^{e}}\label{upboundmin}\\
&\leq& \dfrac{\varepsilon}{4\xi} \|\varphih\|_{2,\Om^{e}}^{2} 
+ \dfrac{\pi^\text{p}_{1}\|D_{\omega}\|_{\infty,\Om^{e}}}{4\xi} 
\| \nabla\varphih\|_{2,\Om^{e}}^{2} + 
\xi\|\delta_{\Om^{e}}^{\frac{1}{2}}\u\nabla\varphih\|_{2,\Om^{e}}^{2}
\end{eqnarray}
Having used Cauchy-Schwarz for the inequality \eqref{hucs1}, then the inverse inequality (see \cite{10.2307/23074327})  
$$\|\Delta\varphih\|_{2,\Om^{e}} \leq \dfrac{C_{inv}}{h_{\Om^{e}}}\|\nabla\varphih\|_{2,\Om^{e}}$$ 
in the argument for \eqref{young}, we also refer to \cite{ciarlet2002finite} and \cite{ern2013theory} for further details on inverse inequalities. Then we used 
\eqref{ubweightedfunction} in \eqref{upboundmin}, and young's product inequality in \eqref{young}. Thus by chosing $\xi=\frac{1}{2}$ we obtain
\begin{equation}\label{upper2}
\left|\int_{\Om^{e}} \left(\varepsilon\varphih^2 - \pi^\text{p}_{1} \nabla\cdot\left({D}_{\omega}\nabla \varphih\right)\right) \left(\delta_{\Om^{e}}^{\frac{1}{2}}\u\nabla\varphih\right) \dw \right|
    \leq
\dfrac{\varepsilon}{2} \|\varphih\|_{2,\Om^{e}}^{2} 
+ \dfrac{\pi^\text{p}_{1}\|D_{\omega}\|_{\infty,\Om^{e}}}{2} 
\| \nabla\varphih\|_{2,\Om^{e}}^{2} + 
\dfrac{1}{2}\|\delta_{\Om^{e}}^{\frac{1}{2}}\u\nabla\varphih\|_{2,\Om^{e}}^{2}
\end{equation}
which we combine together with \eqref{lowerbound2} to obtain 
\begin{eqnarray}
\aeps_{\sc supg}(\varphih,\varphih) &\geq&  \|\varphih\|_{\sc supg}^{2} -\dfrac{1}{2}\sum_{e}^{nel} \left(
\varepsilon \|\varphih\|_{2,\Om^{e}}^{2} 
+ \left(\pi^\text{p}_{1}\|D_{\omega}\|_{\infty,\Om^{e}}\right) 
\| \nabla\varphih\|_{2,\Om^{e}}^{2} + 
\|\delta_{\Om^{e}}^{\frac{1}{2}}\u\nabla\varphih\|_{2,\Om^{e}}^{2}
\right)
\notag\\
&=& \dfrac{1}{2} \|\varphih\|_{\sc supg}^{2} 
\end{eqnarray}
Besides, for the boundedness of the binilinear form we have  
\begin{eqnarray}
\aeps_{\sc supg}(\varphih,\psih) &=&  
 \left(\left( \u\cdot\nabla \varphih\right), \psih \right) 
+ \pi^\text{p}_{1}\left( {D}_{\omega}\nabla \varphih , \nabla\psih \right) + \varepsilon\left(\varphih,\psih\right) \notag
\\&&+\sum_{e}^{nel} 
\int_{\Om^{e}}\delta_{\Om^{e}}\left( \u\cdot\nabla \varphih\right) \left(\u\nabla\psih\right) \dw \notag\\
&& + \sum_{e}^{nel} \int_{\Om^{e}} \left(\varepsilon\varphih\psih + \pi^\text{p}_{1} \nabla\cdot\left({D}_{\omega}\nabla \varphih\right)\right) \left(\delta_{\Om^{e}}^{\frac{1}{2}}\u\nabla\psih\right) \dw \notag\\
&\leq& \left(c_{\Om}\|\u\|_{2} + \pi^\text{p}_{1} \|{D}_{\omega}\|_{\infty} + \varepsilon c_{\Om}^{2} \right) \|\nabla\varphih\|_{2}\|\nabla\psih\|_{2} \notag\\
&&+ \sum_{e}^{nel} \left\|\delta_{\Om^{e}}^{\frac{1}{2}}\u\nabla\varphih \right\|_{2,\Om^{e}} \left\|\delta_{\Om^{e}}^{\frac{1}{2}}\u\nabla\psih \right\|_{2,\Om^{e}} \notag\\
&&+\frac{\varepsilon c_{\Om}^{2}}{2} \|\nabla\varphih\|_{2,\Om^{e}}^{2}\|\nabla\psih\|_{2,\Om^{e}}^{2} + \dfrac{\pi^\text{p}_{1}\|D_{\omega}\|_{\infty,\Om^{e}}}{2} \|\nabla\varphih\|_{2,\Om^{e}}^{2}\|\nabla\psih\|_{2,\Om^{e}}^{2} + \dfrac{1}{2}\|\delta_{\Om^{e}}^{\frac{1}{2}}\u\nabla\varphih\|_{2,\Om^{e}}^{2}\notag\\
&\leq& C \|\nabla\varphih\|_{2}\|\nabla\psih\|_{2} 
\end{eqnarray}
where $C$ is function of $C\left(c_{\Om},\varepsilon,\|{D}_{\omega}\|_{\infty},\|\u\|_{2}\right)$. This completes the proof.
\end{proof}    
Note here that the SUPG method demonstrates a stronger stability property in the streamline direction than the standard Galerkin discretization.
%The existence of the solution can therefore follow by application of the Lax-Milgram theorem. 

Furthermore, the nanoparticle concentration mean value has to be equal to one (as per the dimensionless formulation of the equations). This reflects the conservation of the bulk amount of concentration in the enclosure. This nanoparticle constraint hence writes as
\begin{equation}\label{constraintphi}\int_{\Om}\phi(\omega) \dw=1.\end{equation}
In this light, the handled problem is treated as a variational minimization problem subject to a constraint in the state variable $\phih$ and writes as follows
\begin{equation*}
    \min_{\begin{array}{cc}
         \phih \in H^{1}_{0}(\Om),\\ \int_{\Om}\phih(\omega) \dw=1
    \end{array}}\hspace{-0.6cm}
    \mathcal{J}(\phih,\lambda):=\dfrac{1}{2} \aeps(\phih,\phih) + \lambda \int_{\Om}\phih(\omega) \dw
\end{equation*}
where the constraint \eqref{constraintphi} is token into consideration using the Lagrange multiplier $\lambda$ in the above augmented cost functional. One can easily retrieve the equations that need to be solved for the nanoparticle concentrations as a critical point of $\mathcal{J}(\cdot,\cdot)$, i.e., by deriving with respect to $\phih$ and $\lambda$.

In practice, we assemble the finite element matrix system for the nanoparticle transport equation as follows: 
\begin{equation}\label{LinearSystemFEMnanoparticle}
\begin{bmatrix}
\N  & z\vspace{.05in}\\z^{T} &  0
\end{bmatrix}
\begin{bmatrix}
\phih\\ \lambda
\end{bmatrix}
=\begin{bmatrix}b \\ 1\end{bmatrix}
\end{equation}
where 
\begin{eqnarray*}
\left(\N\right)_{i,j} &=& 
\int_{\Om}\left( \u\cdot\nabla \varphih^{i}\right) \varphih^{j} \dw 
+ \int_{\Om} \pi^\text{p}_{1} D_w\nabla \varphih^{i} \cdot\nabla\varphih^{j} \dw +\int_{\Om} \varepsilon \varphih^{i}\varphih^{j} \dw\\
&&+\sum_{e}^{nel}\int_{\Om^{e}}\delta_{\Om^{e}}\left( \u\cdot\nabla \varphih^{i}\right) \left(\u\nabla\varphih^{j}\right) \dw \\
&&- \sum_{e}^{nel}\int_{\Om^{e}} \delta_{\Om^{e}}\pi^\text{p}_{1} \nabla\cdot\left({D}_{\omega}^{}\nabla \varphih^{i}\right) \left(\u\nabla\varphih^{j}\right) \dw \\
\left(z\right)_{j}&=&  \int_{\Om} 1\cdot\varphih^{j} \dw\\
\left(b\right)_{j}  &=&-\int_{\Om}\pi^\text{p}_{2} {D}_{\TempdH}\nabla\Temph \cdot\nabla\varphih^{j} \dw\\
&&+\sum_{e}^{nel}\int_{\Om^{e}}\delta_{\Om^{e}}\pi^\text{p}_{2} \nabla\cdot\left({D}_{\Temph}\nabla\Temph\right) \left(\u\nabla\varphih^{j}\right) \dw.
\end{eqnarray*}
One may ask whether this augmented matrix \eqref{LinearSystemFEMnanoparticle} is invertible or not. Indeed, it is invertible as per the following proposition 
 \begin{proposition}
 The discrete linear system \eqref{LinearSystemFEMnanoparticle} describing the nanoparticle transport in the enclosure (under the mean value constraint) is consistent. 
 \end{proposition}
 
 \begin{proof}
It is clear that the finite element square matrix $\N:=[c_1|\cdots|c_n]\in\mathbb{R}^{n\times n}$ is non-singular as per Proposition~\ref{propositionWellPosedness}, which means that $0$ is not an eigenvalue of $\N$. Therefore, its set of column vectors $c_i,i=1...,n$ form a basis for the column space Col$(\N)$.
 For any non zero vector $z\in\mathbb{R}^{n}$ with positive entries, there exists a sequence of coefficients $(\alpha)_{j}$ such that we have $z=\sum_{j=1}^{n} \alpha_j c_j$ 
where $\alpha=N^{-1}z$ is clearly unique.

Let us assume that the column vector $\tilde z:=\begin{bmatrix}z\\0\end{bmatrix}$ is also linear combination of columns of the augmented matrix $\begin{bmatrix}\N\\ z^T\end{bmatrix}$. So $\tilde z = \sum_{j} \alpha_{j} \tilde c_{j}$ (where $\tilde c_{j}=[c_j^T,z_j]^T$ , in particular, if we take the bottom line of this linear combination we have $\sum_{j}\alpha_j z_{j}=0$ or equivalently 
$$z^T\alpha = 0$$
which leads to $z^TN^{-1}z=0, \, \forall z\in\mathbb{R}^{n}, \|z\|_2\neq0$. This contradicts $\N$ is non-singular. 
\end{proof}

\subsection{The Algorithm}
 We present our method that combines Newton's method for the resolution of nonlinear PDEs, combining the Momentum and Energy equations, together with the nanoparticle concentration advection dominated at steady state. 
 The physical parameters dependency in the nanoparticle concentration differs from one correlation to another, see for instance~\cite{ho_liu_chang_lin_2010,khanafer_vafai_2017,ASTANINA2018532} without being exhaustive. In most cases a highly nonlinear term involving the nanoparticle concentration appear in the correlation (generally coming from non-linear fitting procedure). In order to avoid any differentiation in the Newton's method with respect to the nanoparticle concentration we will split the resolution method into two and update all variables accordingly in iterative fashion as demonstrated in following Algorithm~\ref{algo}. 
 \bigskip
 
\begin{algorithm}[H]
\SetKwInOut{Parameter}{Parameters}
  \KwIn{$tol,\u^{0},p^{0},\Temp^{0}$}
  \For{$Ra_{nf} \quad in\quad (10^{4}\cdots 10^{8})$} {
  $\epsilon=1$;$k=1$;$\phih^{0}=(\int_{\Om}\dw)^{-1}$\;\emph{Initialization of the nanoparticle concentration}\;
  \While{$\epsilon\geq$ tol}{
    \tcc{Using $\phih^{k-1}$}
     \text{Solve for} $ \begin{bmatrix}\delta\u^{k} ,\delta p^{k} , \delta \Temp^{k}\end{bmatrix}^T$ \text{following Eq.}~\eqref{MomentumEnergydeltaUpdate}\;
     $\epsilon_{m,e} \gets \dfrac{[\delta\uh^{k} ,\delta p_{h}^{k} , \delta \Temp^{k}]^T[\delta\uh^{k} ,\delta p_{h}^{k} , \delta \Temp^{k}]}{[\uh^{k},p_{h}^{k},\Temp^{k}]^T[\uh^{k},p_{h}^{k},\Temph^{k}]}$\;
     \text{Update} $\begin{bmatrix}\uh^{k+1} ,  p_{h}^{k+1} , \Temph^{k+1}\end{bmatrix}^T$ \text{following  Eq.}~\eqref{MomentumEnergyUpdate}\;
     \text{Solve for} $\phih^{k+1}$  \text{under the mean constraint Eq.}~\eqref{LinearSystemFEMnanoparticle}\;
     \tcc{Nanoparticle concentration with the SUPG stabilization}
     $k\gets k+1$\;
  }
   $k\gets 0$\;
  }
\caption{Combined Newton’s method and SUPG for nanofluid equation\label{algo}}
\end{algorithm}
\bigskip
As it is showing in Algorithm~\ref{algo} a split procedure is adopted, where the tangent equation only concern the momentum and energy equations leading to an update, through Newton’s method, of the velocity and temperature respectively.
Although, these later two equations are dependent upon the nanoparticle concentration through the mixture fluid density,
viscosity and thermal conductivity. The proposed method assumes that the velocity and the temperature are constant through one iteration of Newton’s method. Their update will then be ensured within the next iteration after an exact resolution (through LU decomposition) of the nanoparticle concentration (with SUPG) based on the new variables of the velocity and temperature coming out of the previous Newton’s iteration. This alternating combination is applied throughout the iterations and leads to convergence for all variables involved. The split approach has lead to a simple yet effective implementation of Newton’s method involving highly non-linear parameters. Therefore, it skipped the tedious calculation of the tangent equations of the whole coupled system of four PDEs. Furthermore, less memory storage is then deployed hence rapid calculation. The above procedure is thus repeated with a predefined
set of increasing Rayleigh numbers.
%~~~~~~~~~~~~~~~~~~~~~~~~~~~~~~~~
\section{Numerical experiments and validations}\label{NumRes}
%~~~~~~~~~~~~~~~~~~~~~~~~~~~~~~~~
In this section, we investigate the numerical treatment of the heat transfer enhancement in a differentially heated enclosure using variable thermal conductivity and variable viscosity of the alumina-water nanofluid (Al$_2$O$_3$-water). The validation of numerical scheme is done through two processes. The first focuses on the validation of the numerical scheme relative to the resolution of the momentum and energy equations regardless of the volume fraction. In this case we consider the pure water heat transfer calculation in a square cavity. The second considers the comparison of the present numerical results with available numerical and experimental data. 

\subsection{Numerical schemes validations}
\begin{figure}[!hptb]
\centering
\includegraphics[width=6cm,height=6cm]{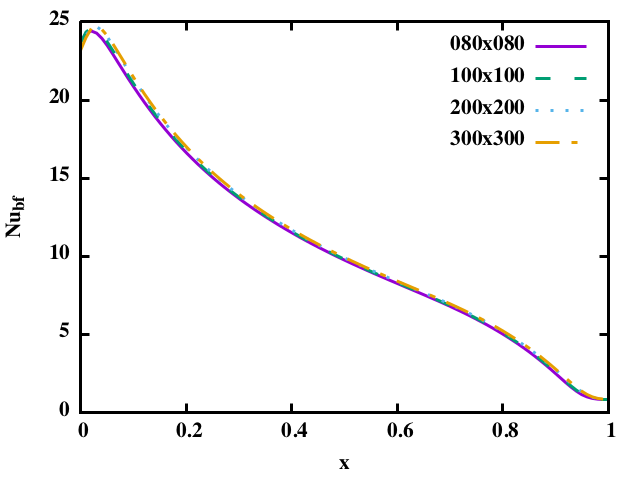}
\includegraphics[width=6cm,height=6cm]{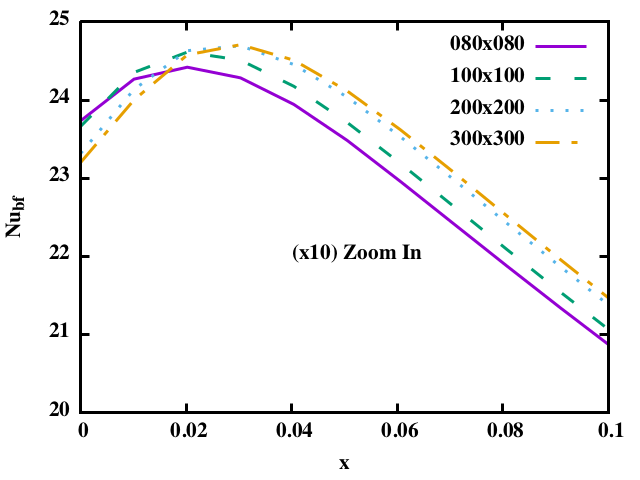}
\caption{Mesh sensitivity; Results present Nusselt (base fluid) number plotted along the Heated wall of the cavity.}\label{MeshCalibrationNusselt}
\end{figure}

We present in Figure \ref{MeshCalibrationNusselt} the mesh sensitivity results for the base fluid (only the Newton's solver). 

\begin{figure}[!hptb]
\centering
\includegraphics[width=6cm,height=6cm]{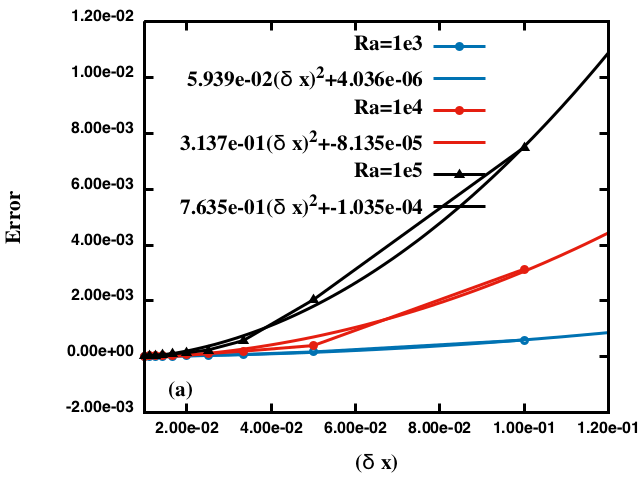}
\includegraphics[width=6cm,height=6cm]{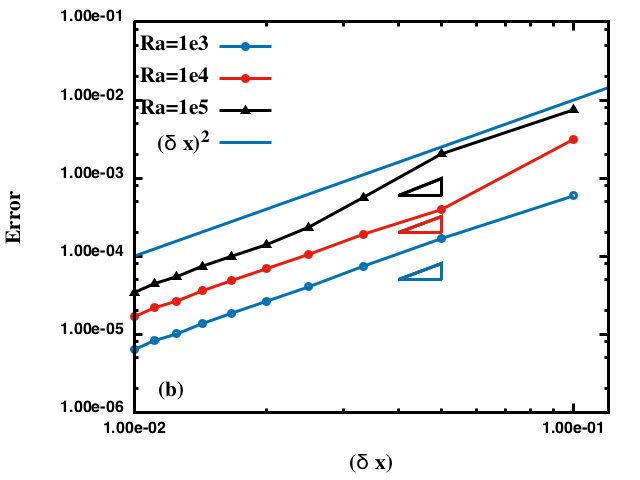}
\caption{Order of Convergence, with respect to the mesh size, of the presented scheme for the nanoparticle concentration with SUPG stabilization.}\label{OrderOfConvergenceNanoParticleEquation}
\end{figure}
Figure\ref{OrderOfConvergenceNanoParticleEquation} shows the quadratic convergence of the proposed numerical scheme with the use of the SUPG stabilization technique. In addition, as we shall explain in the sequel, we enforced the bulk of the nanoparticle concentration to be of mean value equal to $1$, through minimization under constraint problem. A continuous $\textbf{P}2$ finite element was used for the advection dominated nanoparticle concentration equation. Results show that the numerical scheme with the SUPG is stable and satisfies the convergence property of the FE discretization~\cite{BENEDETTO201618} even at high Rayleigh number for a turbulent flow~\cite{WERVAECKE2012109}. 

In Figure~\ref{figwithwithoutSUPG} we present the benefit of the implementation of the SUPG method in order to eliminate numerical artifacts showing up in the concentration of the nanoparticles, which is governed by an advection dominated problem. 
 \begin{figure}[!htpb]
 \centering
 \begin{tabular}{lccc}
 
  &{\bf$\Ra_\text{nf}$=1.E+06}&{\bf$\Ra_\text{nf}$=1.E+07}&{\bf$\Ra_\text{nf}$=1.E+08}\\
 \rotatebox{90}{\scriptsize\hspace{.5in}{\bf\large Without SUPG}}  
  &\includegraphics[width=4.5cm,height=4.5cm]{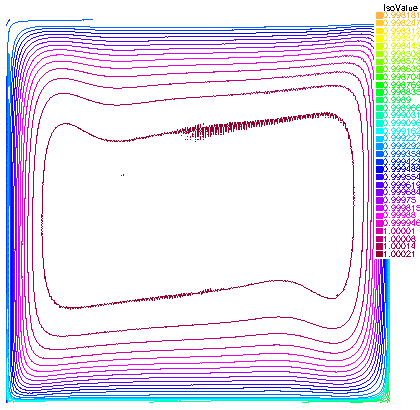} 
  &\includegraphics[width=4.5cm,height=4.5cm]{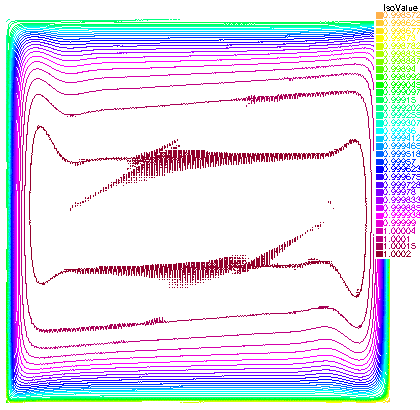} &\includegraphics[width=4.5cm,height=4.5cm]{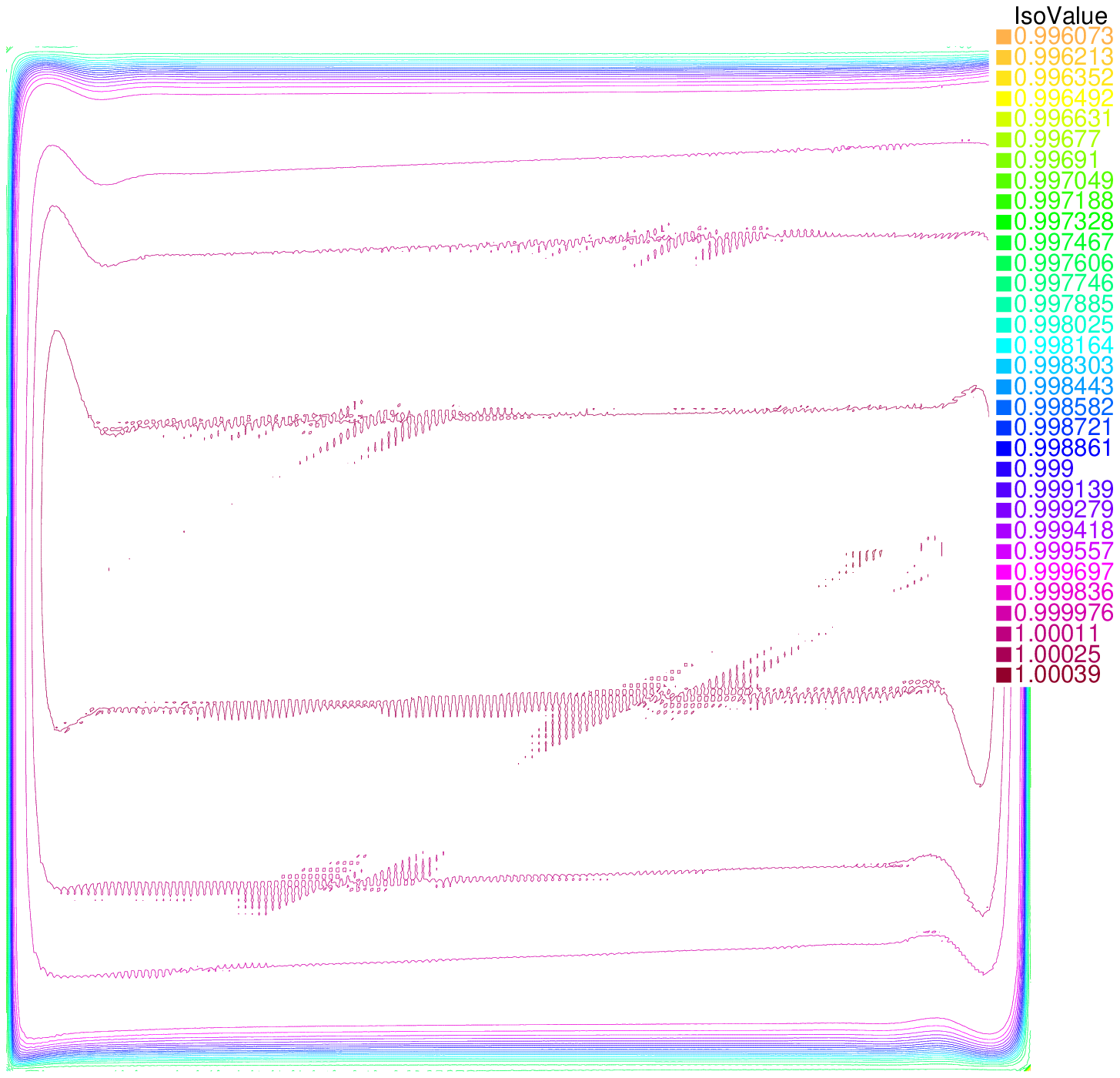}\\
 \rotatebox{90}{\scriptsize\hspace{.5in}{\bf\large With SUPG}}  
  &\includegraphics[width=4.5cm,height=4.5cm]{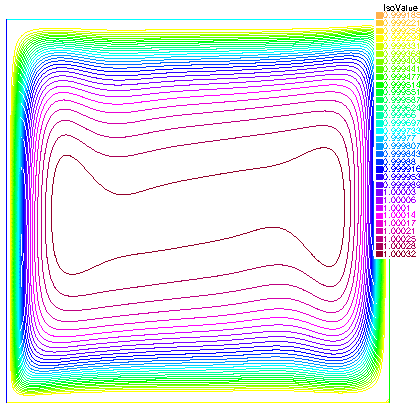}
  &\includegraphics[width=4.5cm,height=4.5cm]{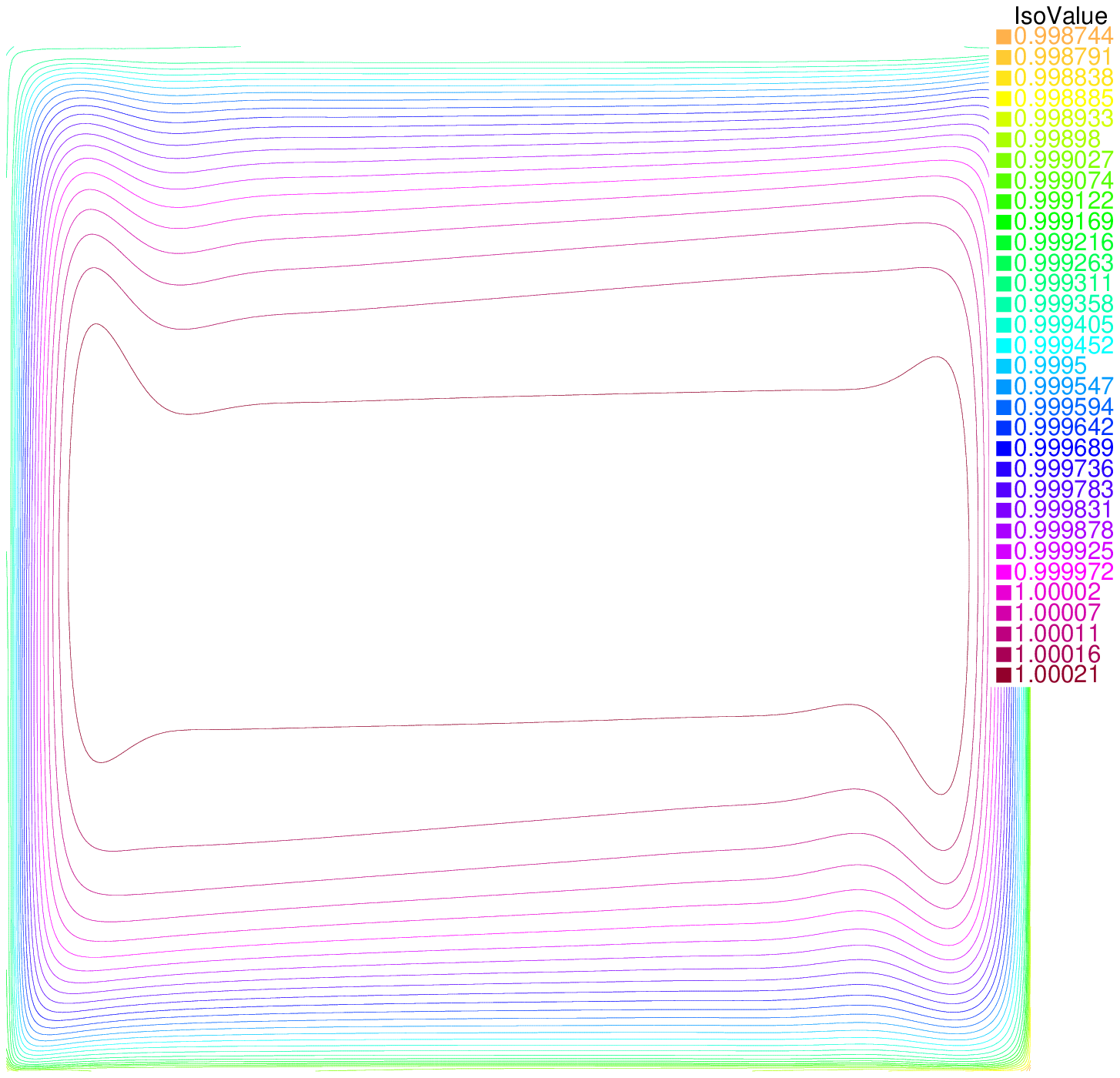}  
  &\includegraphics[width=4.5cm,height=4.5cm]{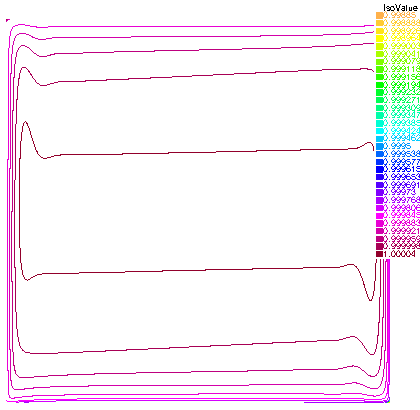}
 \end{tabular}
 \caption{Stabilizatoin effect of the SUPG on the nanoparticle solutions by elimination of numerical artifacts. Plot of isovalues of the concentration in the case of $\phi_{0}=3\%$}\label{figwithwithoutSUPG}
 \end{figure}
 It is clearly shown in Figure~\ref{figwithwithoutSUPG} that in the case without SUPG stabilisation the numerical spurious are more accentuating for high Rayleigh simulations, where indeed, high Peclet number takes place.  
 
\subsection{Validation -vs- Experimental results}
Our numerical scheme is validated using detailed comparison with the experimental data of Ho {\it et al.} in \cite{ho_liu_chang_lin_2010} using Al$_2$O$_3$ which their thermophysical properties are reported in Table~\ref{tab:physical_properties}. In their experimental investigations, the authors studied the heat transfer characteristics of alumina-water nanofluid enclosed in square cells. They used three different cell geometries and different heating conditions ($\TempdH-\TempdC$) to increase the Rayleigh number. Cases with $0, 1$ and $3\%$ of nanoparticle concentration were considered. It has been shown recently \cite{alosious2017experimental} that numerical experiments using continuous models over-predict the heat transfer enhancement. 
\begin{table}[h]
    \centering
    \begin{tabular}{lcc} \hline\\
    \text{Physical properties} & \text{Base fluid}  & Al$_2$O$_3$  \\ \hline
        $c_\text{p}$ (Jkg$^{-1}$ K$^{-1}$) & 4179 &  765\\
        $\rho$ (kg m$^{-3}$) & 997.1&3970\\
        $k$ (W m$^{-1}$K$^{-1}$) & 0.613&25\\
        $d_\text{p}$ (nm) &0.384 &47\\
        $\alpha\cdot 10^{-7}$ (m$^2$s$^{-1}$) &1.47 &82.23\\
        $\beta\cdot 10^{-5}$ (K$^{-1}$) & 21&0.85\\ \hline
    \end{tabular}
    \caption{Physical properties of base fluid and Al$_2$O$_3$ nanoparticles}
    \label{tab:physical_properties}
\end{table}
\begin{figure}[!htbp]
\begin{tabular}{ccc}
\includegraphics[width=5.2cm,height=5.4cm]{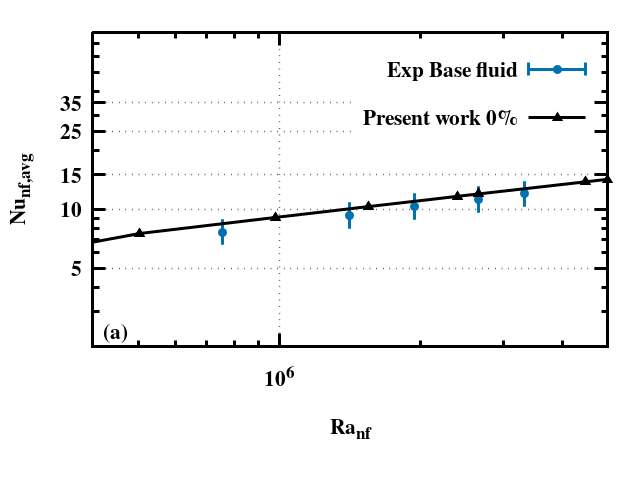}
\includegraphics[width=5.2cm,height=5.4cm]{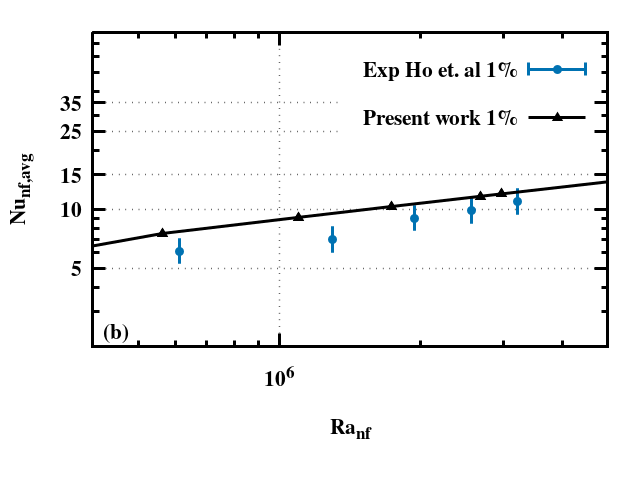}
\includegraphics[width=5.2cm,height=5.4cm]{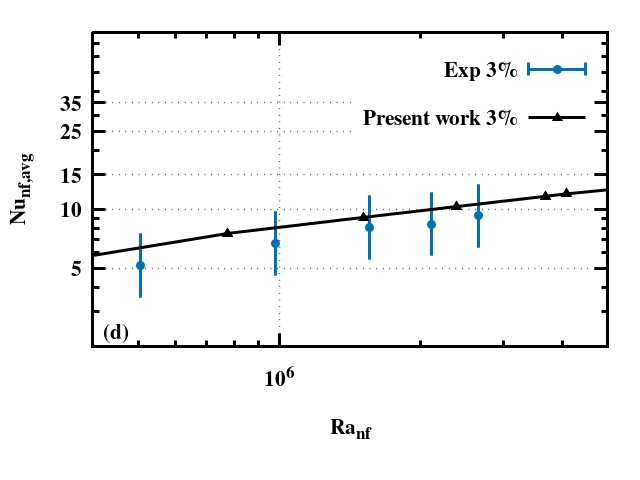}
\end{tabular}
\caption{Comparison of the numerical experiments with the experimental results of Ho \textit{et al.}\cite{ho_liu_chang_lin_2010} for the first cell case}\label{figcell1}
\end{figure}

Our validation and comparison of the numerical results against experimental finding of nanofluid heat transfer followed the cells presentation as in \cite{ho_liu_chang_lin_2010}. 
% This allow validation in different regimes such as i) laminar with the first cell, ii) transition to turbulent with the second cell and iii) turbulent regime with the third cell. 
The corresponding results are reported in Figure~\ref{figcell1},\ref{figcell2} and \ref{figcell3}. In each of these figures we plot separately the cases of base fluid (left) the nanofluid concentration of $1\%$ (middle) and the nanofluid concentration of $3\%$ (left). To produce these results, we have used viscosity and thermal conductivity correlations as reported by Ho \textit{et al.} in \cite{ho_liu_chang_lin_2010} as follows.
\begin{eqnarray*}
\mathcal{C}_{\mu}(\phi) &=& 1 + 4.93 \cdot\phi + 222.4 \cdot \phi^2\\
\mathcal{C}_{k}(\phi) &=& (1+2.944 \cdot\phi + 19.672 \cdot \phi^2), 
\end{eqnarray*}
standing for the viscosity and thermal conductivity respectively.
% In the above $\Pr_{\Tempd}=\dfrac{\mu_\text{bf}(\Tempd)}{\rho_\text{bf}\alpha}$ and $\Re_{\Tempd}=\dfrac{\rho_\text{bf}k_\text{b}}{3\pi\mu_\text{bf}^{2}(\Tempd)l_\text{bf}}\Tempd$, where $l_\text{bf}=0.17nm$ stands for the mean-path of the fluid particles.

\begin{figure}[!hptb]
\begin{tabular}{c}
\includegraphics[width=5.2cm,height=5.4cm]{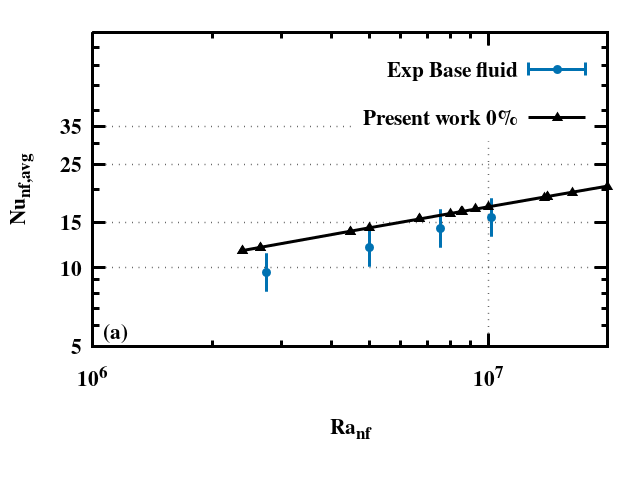}
\includegraphics[width=5.2cm,height=5.4cm]{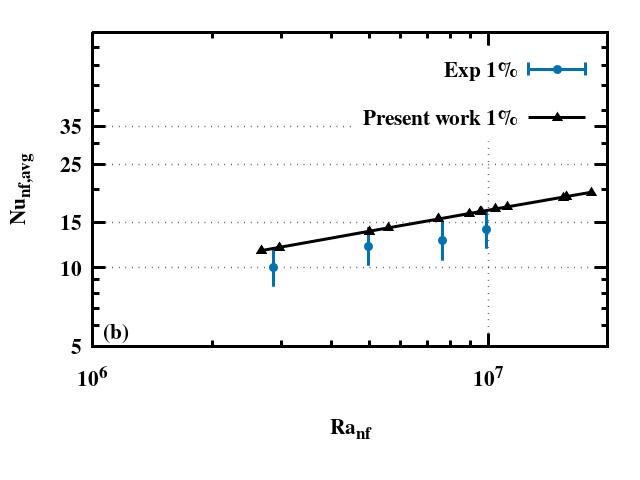} 
\includegraphics[width=5.2cm,height=5.4cm]{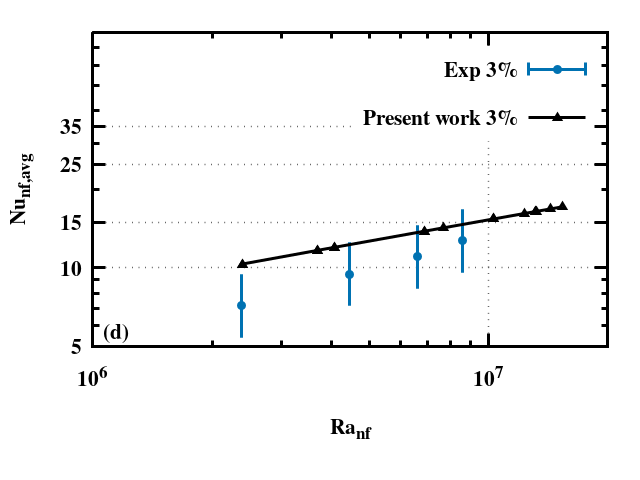}
\end{tabular}
\caption{Comparison of the numerical experiment with the experimental results of Ho \textit{et al.} \cite{ho_liu_chang_lin_2010} for the second cell case}\label{figcell2}
\end{figure}

\begin{figure}[!hptb]
\begin{tabular}{c}
 \includegraphics[width=5.2cm,height=5.4cm]{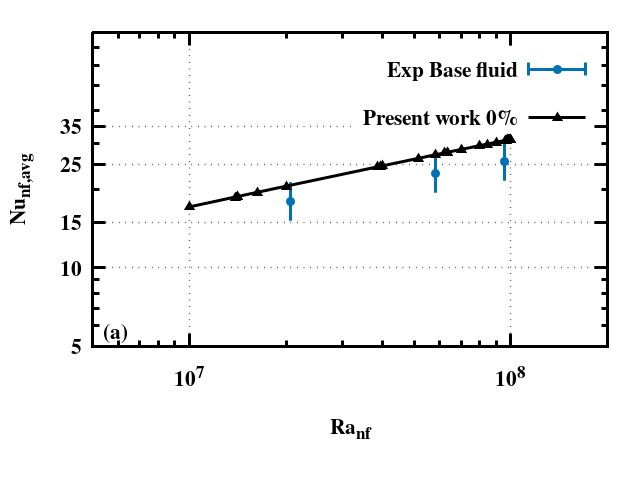}
\includegraphics[width=5.2cm,height=5.4cm]{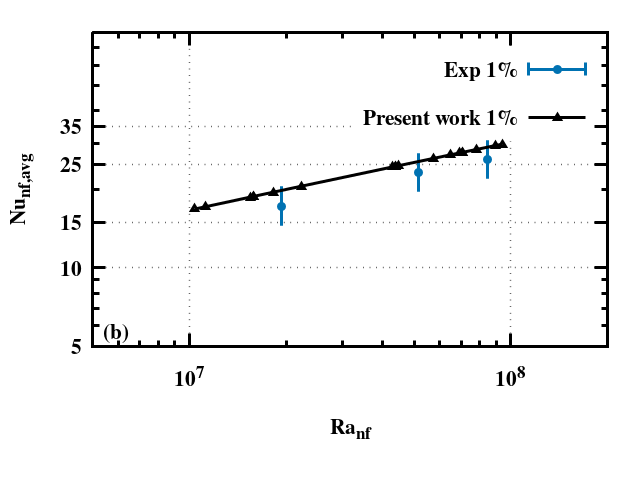}
 \includegraphics[width=5.2cm,height=5.4cm]{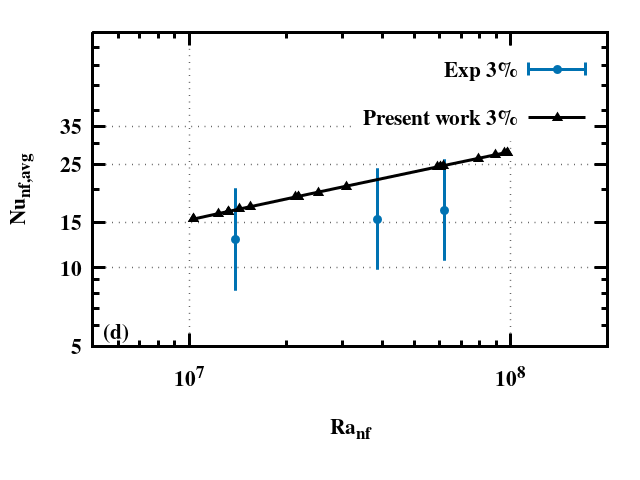} 
\end{tabular}
\caption{Comparison of the numerical experiment with the experimental results of Ho \textit{et al.} \cite{ho_liu_chang_lin_2010} for the third cell case}\label{figcell3}
\end{figure}

Based on the averaged Nusselt number values of the nanofluid, the present predictions show a reasonably good agreement with the experimental data for low and moderate Ra (Figures \ref{figcell1} and \ref{figcell2}). Whereas, for high Ra and high concentration, the numerical calculations are found to overestimate the Nusselt number. In fact, for the high $\Ra_\text{nf}$ number cases (Figure \ref{figcell3}), the experimental results show a more pronounced Nusselt number deterioration for the cases with nanofluid in comparison to the one with the base fluid. In their paper, Ho et al. \cite{ho_liu_chang_lin_2010} suspected that this behavior could be attributed to the transport mechanisms associated with nanoparticle-fluid interactions such as Brownian diffusion and thermophoresis in addition to the impact of the thermophysical properties changes. The present numerical predictions, however, indicate that the Buongiorno model which is intended to specifically account for these two mechanisms, is not actually able to mimic the equivalent heat transfer impairment. This might be suggesting that perhaps additional forces should be incorporated in the Buongiorno transport equations to provide a better physical model for the nanoparticles concentration, capable of reflecting the effect of the nanoparticles on the heat transfer impairment observed in the experimental data.

% %~~~~~~~~~~~~~~~~~~~~~~~~~~~~~~~~
 \subsection{FE stabilized Buongiorno model -vs- multi-phase model}
% %~~~~~~~~~~~~~~~~~~~~~~~~~~~~~~~~
This subsection is devoted to the comparison of the numerical results obtained by Algorithm \ref{algo} based on FEM discretization of the Buongiorno nanofluid transport model, against results of \cite{chen_wang_liu_2016} that are based on finite volume discretization using Fluent~\cite{ansys2016ansys} (commercial software). Here it is worth noticing that both methods deal with same physics of nanofluid transport, although, use different equations models. Indeed, the aforementioned results are based on a solid-liquid mixture model which solves the momentum equations with an additional term to account for the phases drift velocity, the continuity equation, and energy equation for the mixture. The model adopts algebraic expressions for the relative velocities which are then used to define the drift velocities (see Fluent's documentation for more details~\cite{ansys2016ansys}). Both numerical results are compared to experimental results of Ho {\it et al.}\cite{ho_liu_chang_lin_2010}.   

\begin{figure}[!hptb]
\centering
\begin{tabular}{cc}
\includegraphics[width=5.2cm,height=5.4cm]{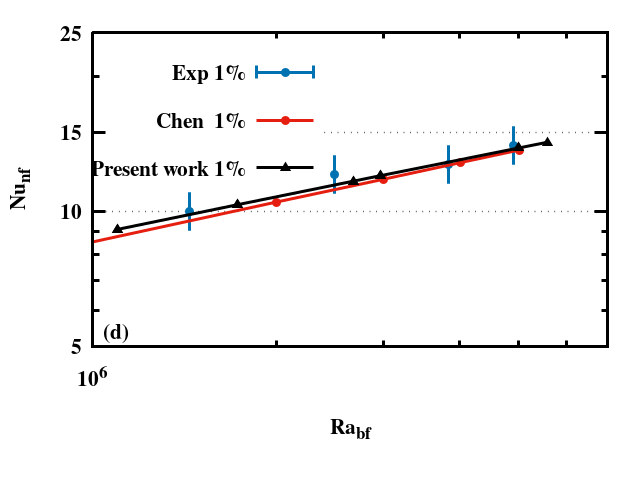}
&
\includegraphics[width=5.2cm,height=5.4cm]{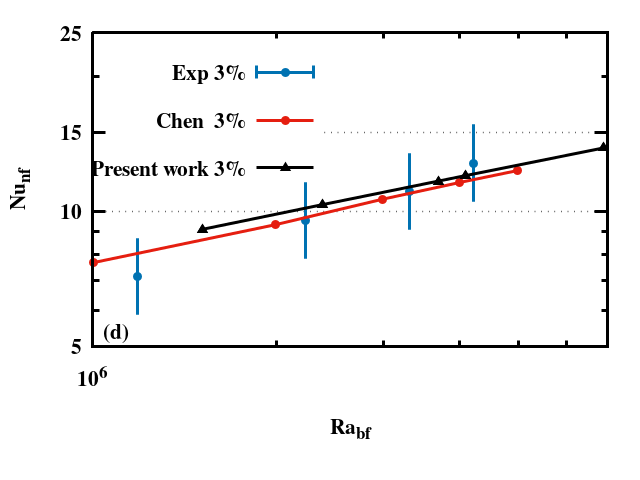}
\end{tabular}
\caption{Comparison of the numerical experiments with the numerical experiment of Chen \textit{et al.}\cite{chen_wang_liu_2016} for the range of $Ra$ in base fluid formulation.}
\label{Num-vs-chen}
\end{figure}
Figure~\ref{Num-vs-chen} depicts both numerical results and showcase a good agreement between the two methods. The range of data available for this validation and comparison is $1\cdot10^{6}$ to $6\cdot10^{6}$. For the case of $1\%$ of the nanoparticle bulk concentration, Buongiorno model seems to have better prediction of the heat transfer in term of the Nussult number of the nanofluid. However, this advantage becomes marginally on the side of the multi-phase model while we increase the bulk of nanoparticle concentration to $3\%$.  

\begin{figure}[!htbp]
\begin{tabular}{ccccc}
\rotatebox{90}{\scriptsize\hspace{.5in}$\Ra_\text{nf}$=1.E4}&
\includegraphics[width=4cm,height=4cm]{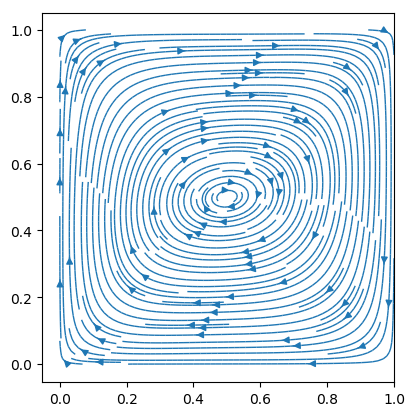}&
\includegraphics[width=4cm,height=4cm]{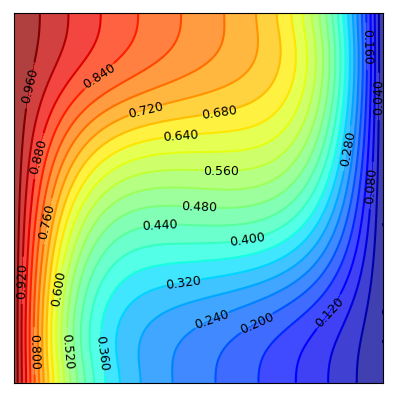}&
\includegraphics[width=4cm,height=4cm]{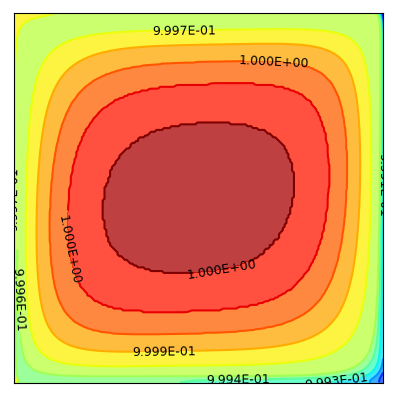}\\
 % ~ ~ ~~ ~ ~ ~ ~~ ~ ~ ~~ ~ ~ 
\rotatebox{90}{\scriptsize\hspace{.5in}$\Ra_\text{nf}$=1.E5}&
\includegraphics[width=4cm,height=4cm]{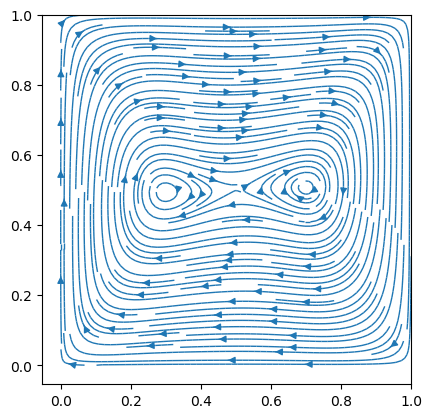}&
\includegraphics[width=4cm,height=4cm]{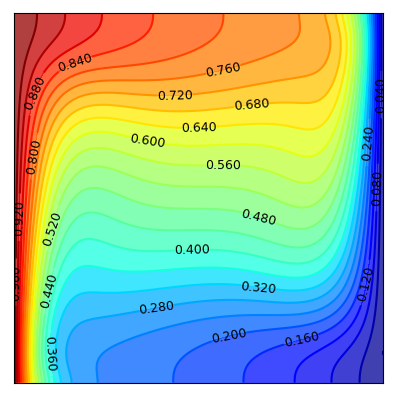}&
\includegraphics[width=4cm,height=4cm]{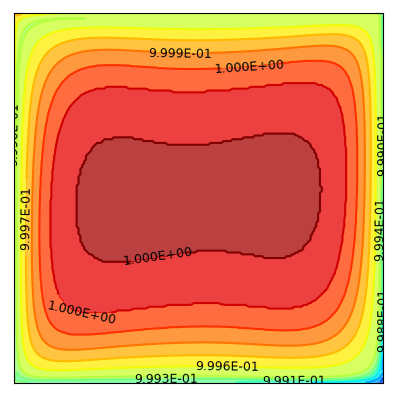}\\
% ~ ~ ~~ ~ ~ ~ ~~ ~ ~ ~~ ~ ~ 
\rotatebox{90}{\scriptsize\hspace{.5in}$\Ra_\text{nf}$=1.E6}&
\includegraphics[width=4cm,height=4cm]{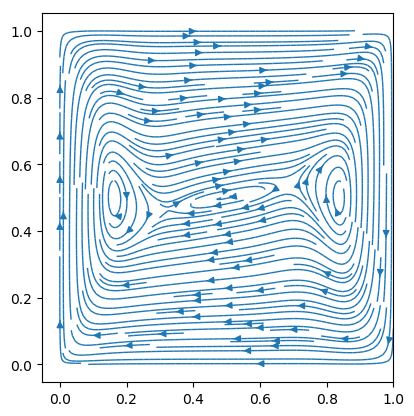}&
\includegraphics[width=4cm,height=4cm]{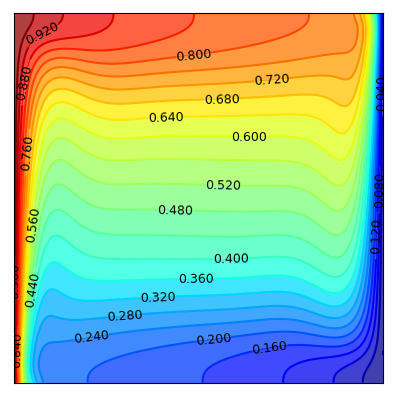}&
\includegraphics[width=4cm,height=4cm]{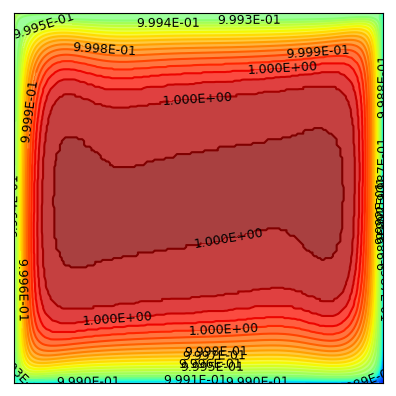}\\
% ~ ~ ~~ ~ ~ ~ ~~ ~ ~ ~~ ~ ~ 
\rotatebox{90}{\scriptsize\hspace{.5in}$\Ra_\text{nf}$=1.E7}&
\includegraphics[width=4cm,height=4cm]{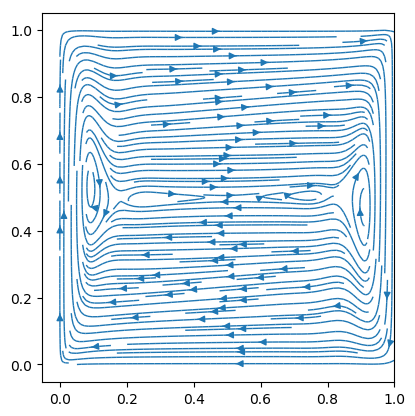}&
\includegraphics[width=4cm,height=4cm]{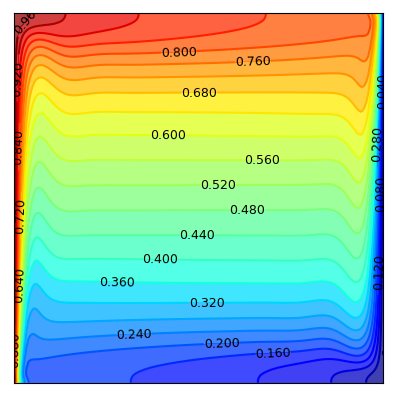}&
\includegraphics[width=4cm,height=4cm]{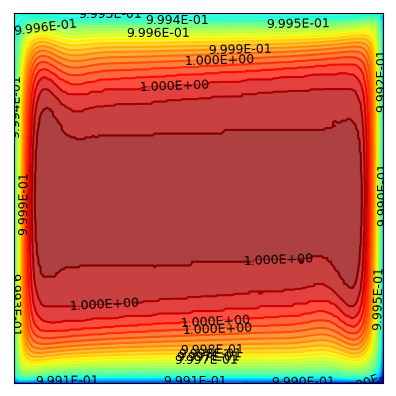}\\

% ~ ~ ~~ ~ ~ ~ ~~ ~ ~ ~~ ~ ~ 
\rotatebox{90}{\scriptsize\hspace{.5in}$\Ra_\text{nf}$=1.E8}&
\includegraphics[width=4cm,height=4cm]{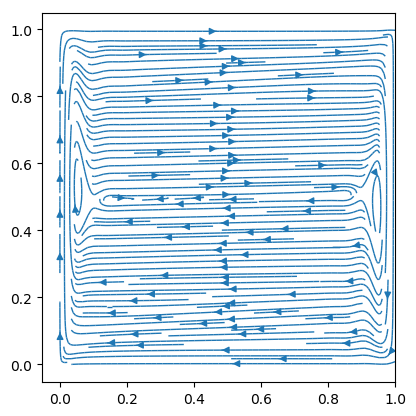}&
\includegraphics[width=4cm,height=4cm]{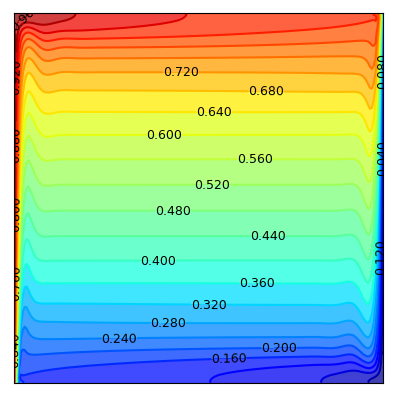}&
\includegraphics[width=4cm,height=4cm]{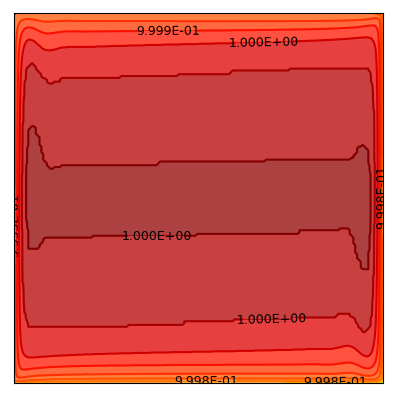}\\
% ~ ~ ~~ ~ ~ ~ ~~ ~ ~ ~~ ~ ~
\end{tabular}
\caption{(From left to right) velocity streamlines, heat distribution and nanoparticle concentration distribution. Plots, show (from top to bottom) the effect of increasing the Rayleigh number on the profile of the variables listed earlier. Numerical simulation was performed through finite element discretization, using Lagrange polynomials of degree two for the velocity and of degree 1 for the temperature and nanoparticle concentration respectively.}\label{StreamContourCencentration}
\end{figure}

Plots of streamlines contours, isothermal lines and nanoparticle concentration distribution are shown in Figure~\ref{StreamContourCencentration}, in which we vary the Rayleigh number $\Ra_\text{nf}$. These results show that the stream-lines exhibit recirculations that get flatten with the increase of the Rayleigh number. These recirculations get localized in the middle of the cavity and near the hot and cold walls. Whereas, isothermal lines become more and more horizontal which lead to high temperature gradient near the hot and cold walls. One, here, can directly see the increase of the heat transfer with the increase of the Rayleigh number. Besides, both the energy and the nanoparticle concentration are advected mainly by the buoyancy driven flow. Although, the energy equation, has a considerable contribution of the thermal diffusion compared to the nanoparticle transport equation, which has much smaller diffusion terms. Indeed, the Brownian diffusivity is very weak in comparison to the thermal diffusion and even negligible when compared to the advection term. This, in fact, would drastically affect the numerical approximation of this advection dominated equation. The SUPG artificial viscosity through the streamlines is, therefore, needed to stabilize the numerical scheme for high Rayleigh/Peclet number as shown in Figure~\ref{figwithwithoutSUPG}.  As the nanoparticle transport equation is mainly driven by the advection, one would expect the distribution of the concentration to be very similar to the stream-line of the flow. Indeed, this behaviour is observed in the third column of the plot in Figure~\ref{StreamContourCencentration}.
\begin{figure}[hptb]
\centering
\begin{tabular}{cc}
\includegraphics[width=5.5cm,height=5.5cm]{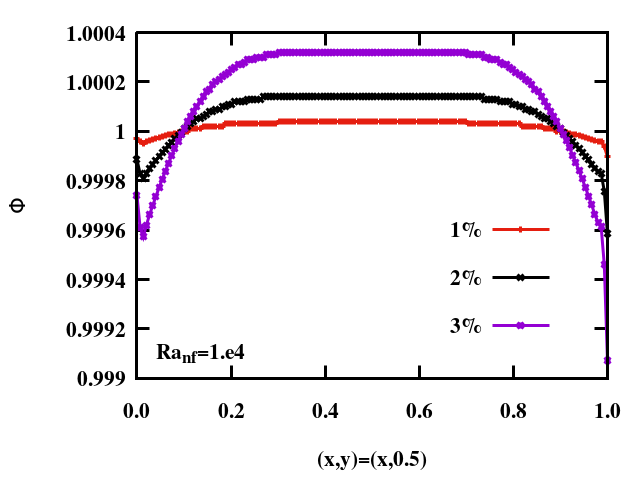}&
\includegraphics[width=5.5cm,height=5.5cm]{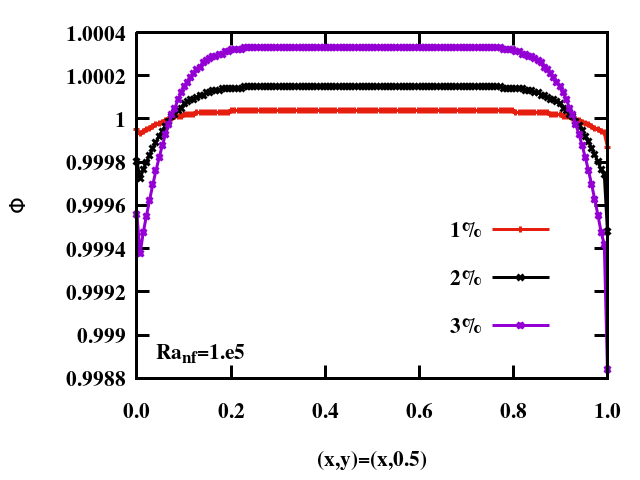}\\
\includegraphics[width=5.5cm,height=5.5cm]{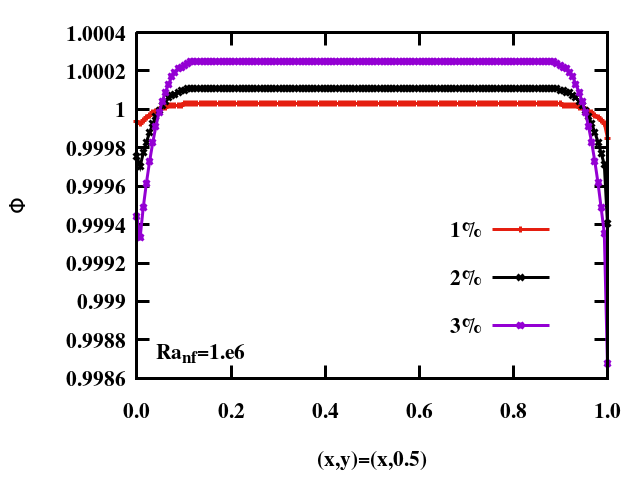}&
\includegraphics[width=5.5cm,height=5.5cm]{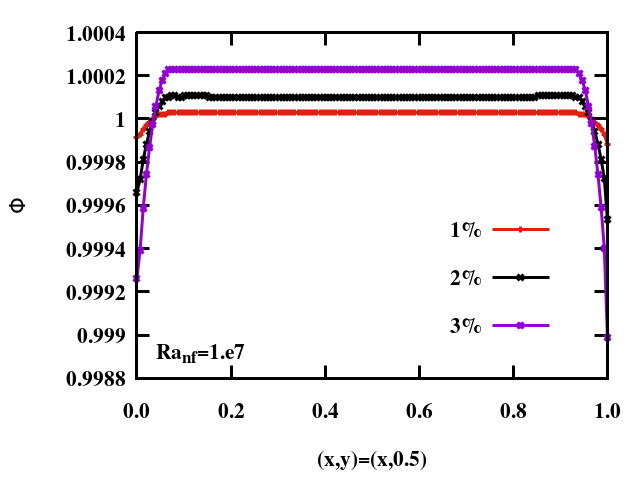}
\end{tabular}
\caption{Nanoparticle concentration profile along the $(x,\frac{1}{2})$ varying from $\%0$ to $\%3$.}\label{figPHIprofile}
\end{figure}
 
 Figure~\ref{figPHIprofile} displays the nanoparticles concentration profile along the line $(x,0.5)$ for different Rayleigh numbers and averaged concentration $\phi$  values ranging from $1\%$ to $3\%$. Although, the profiles exhibit only minor variation, of the order of $10^{-4}$ in magnitude along the horizontal line, very interesting phenomena occurring in the vicinity of the hot and cold walls can be observed. Near the cold wall, for instance, the concentration profiles decrease sharply as the nanoparticles get closer to the cold wall. The slop of this decrease is a function of both; the averaged nanoparticles concentration value and the Rayleigh number. If one is to recall the thermophoresis effect which tends to move particles from hotter to colder zones, then even in the near-the-wall zone in which only small convective velocity magnitude exist, this effect is not dominant and the nanoparticles are still being carried away by the recirculating motion of the carrier fluid. On the hot wall, however, a different picture is depicted.  Although it would not be a straight forward task to pin point the dominant term in this zone, the resulting force seems to favor a higher nanoparticles concentration at the wall vicinity followed by a sharper increase which can be translated to the fact that the convective term regains its dominant effect further away from the wall. All these nanoparticles that are removed from the walls region, by one of the two mechanisms described above, get accumulated in the center resulting in a relatively higher nanoparticles concentration. Unfortunately, as stated above, the concentration variation along this line remains marginal, hence, one would not expect it to provide a dramatically different outcomes from running the simulations with a constant concentration value, hence, assuming a single-phase model approach.

%~~~~~~~~~~~~~~~~~~~~~~~~~~~~~~~~
\section{Conclusion}
%~~~~~~~~~~~~~~~~~~~~~~~~~~~~~~~~
We presented in this article a numerical technique based on Newton-Raphson iterations to solve the nanofluid heat transfer problem in a square cavity with variable properties. In addition to its generality (regardless of the correlation used for the variable properties), our technique has mainly two advantages compared to the conventional use of Newton's iterations:
\begin{itemize}
    \item Firstly, it avoids the difficulty coming from the highly non-linear dependency upon the concentration in several correlations published in the literature. Indeed, the Jacobian (tangent problem) disregards nanoparticle concentration and only considers the velocity, pressure, and temperature, while nanoparticle concentration gets updated iteratively. Admittedly, the momentum and energy equations are solved through Newton's iterations, as the dominant (Navier-Stokes) equations is quadratic for the velocity variable, while the nanoparticle transport equation gets solved right after each iteration of the momentum and energy equations.
    \item Secondly, the proposed split leads to less memory consumption and allows the viscosity, density, and thermal conductivity of recirculating flow to be updated at each Newton's iteration.
\end{itemize}  
The numerical experiments based on the Finite Element discretization of the nanofluid heat transfer problem have been regularized using the SUPG method, which showed to be very effective in stabilizing the numerical solution by wiping the spurious oscillations without wrecking the solution. Here in particular we found that the ratio formula between the local Peclet number and global Raleigh number is a good combination for the regularization function used in the SUPG. Besides, our numerical scheme has been thoroughly validated against experimental results and showed a good agreement over a large spectrum of Rayleigh numbers ranging from $10^4$ to $10^8$. 
The present study also reveals that although the Buongiorno’s four equations-based nanofluid transport model, tested herein, is able to capture additional physical phenomena affecting the nanoparticles distribution, additional forces might have to be accounted for within this model in order to mimic heat transfer deterioration of similar amplitude as it is observed in the experimental data.               

%~~~~~~~~~~~~~~~~~~~~~~~~~~~~~~~~
%\section*{Acknowledgment}
%~~~~~~~~~~~~~~~~~~~~~~~~~~~~~~~~
\begin{appendices}
\section{Density dimensionless derivation}
Following Eq.\eqref{rhonf} we have
$$
\rhonf = (1-\phidim) \rhobf+\phidim \rhonp
$$
note also from Eq.~\eqref{rhobetanf},
$$
(\rho\beta)_\text{nf}  = (\rho\beta)_\text{bf}(1-\phidim) + (\rho\beta)_\text{np}\phidim.
$$
where $(\rho\beta)_\text{bf}=\rho_\text{bf}\beta_\text{bf}$ and $(\rho\beta)_\text{np}=\rho_\text{np}\beta_\text{np}$. Hence 
\begin{eqnarray*}
\dfrac{(\rho\beta)_\text{nf} }{\rhonf} &=& 
  \dfrac{\rho_\text{bf}\beta_\text{bf}(1-\phidim)}{ (1-\phidim) \rhobf+\phidim \rhonp}
+ \dfrac{ \rho_\text{np}\beta_\text{np}\phidim}{ (1-\phidim) \rhobf+\phidim \rhonp}\\
&=&   \dfrac{\beta_\text{bf}(1-\phidim)}{ (1-\phidim)+\phidim \dfrac{\rhonp}{\rhobf}}
+ \dfrac{ \phidim}{ (1-\phidim) +\phidim \dfrac{\rhonp}{\rhobf}}\bnp\\
&=& \bbf \left( \dfrac{(1-\phidim)}{ (1-\phidim)+\phidim \dfrac{\rhonp}{\rhobf}}
+ \dfrac{ \phidim}{ (1-\phidim) +\phidim \dfrac{\rhonp}{\rhobf}}\dfrac{\bnp}{\bbf} \right)\\
\end{eqnarray*}
Let 
$$\M:=\left( \dfrac{(1-\phidim)}{ (1-\phidim)+\phidim \dfrac{\rhonp}{\rhobf}}
+ \dfrac{ \phidim}{ (1-\phidim) +\phidim \dfrac{\rhonp}{\rhobf}}\dfrac{\bnp}{\bbf} \right)$$

\bigskip

% \begin{eqnarray*}
% \dfrac{\bnf\rhobf}{\rhonf} &=& \dfrac{1-\phi}{1-\phi+\phi\left(\dfrac{\rhonp}{\rhobf}\right)^2}\bbf +\dfrac{\K\phi}{1-\phi+\phi\dfrac{\rhonp}{\rhobf}}\bnp\\
% &=& \bbf \left( \dfrac{1-\phi}{1-\phi+\phi\left(\dfrac{\rhonp}{\rhobf}\right)^2}\bbf +\dfrac{\K\phi}{1-\phi+\phi\dfrac{\rhonp}{\rhobf}}\dfrac{\bnp}{\bbf} \right)
% \end{eqnarray*}
% note that 
% \begin{equation}\label{Mfi}
% \dfrac{\bnf\rhobf}{\rhonf} :=\bbf\M(\phi)
% \end{equation}
  
\section{Momentum equation}
\begin{equation}
\left(\udim\cdot\nabladim\right)\udim  = \dfrac{-1}{\rhonf}\nabladim\pdim + \dfrac{1}{\rhonf} \nabladim\cdot \left(\munfdim \left( \nabladim\udim + (\nabladim\udim)^{t}\right) \right)
								+ \dfrac{\g}{\rhonf}(\rho_{\infty}-\rho_{c})
\end{equation}

Moving toward dimensionless variables the above equation writes
\begin{equation}
\dfrac{\alpha^2}{L^3} \left(\u\cdot\nabla\right)\u  = \dfrac{-\rhobf\alpha^2}{L^{2}\rhonf}\nabladim\pdim 
									+ \dfrac{\alpha\mubf}{L^3\rhonf} \nabladim\cdot \left(\munf \left( \nabladim\udim + (\nabladim\udim)^{t}\right) \right)
									+ \dfrac{\g\bnf}{\rhonf}\rho_{\infty}(\Temp_{h}-\Temp_{c})\Temp
\end{equation}
Multiplying the above equation by $\dfrac{L^3}{\alpha^2}$ we obtain
\begin{eqnarray*}
 \left(\u\cdot\nabla\right)\u  &=& \dfrac{\rhobf}{\rhonf}  \nabla p + \dfrac{\mubf}{\alpha\rhobf\left(1-\phi+\phi\dfrac{\rhonp}{\rhobf}\right)} \nabla\cdot \left(\munf \left( \nabla\u+ (\nabla\u)^{t}\right) \right)\\
 &&+\dfrac{\M L^3\g\bnf}{\alpha^2\rhonf}\rho_{\infty}(\Temp_{h}-\Temp_{c})\Temp
\end{eqnarray*}
which rewrites using the non-dimensional constants as follows
\begin{eqnarray}
 \left(\u\cdot\nabla\right)\u  &=& \pi^{m}_{1}(\phi) \nabla p + \pi^{m}_{2}(\phi) \nabla\cdot \left(\munf \left( \nabla\u+ (\nabla\u)^{t}\right) \right) +\pi^{m}_{3}(\phi)\Temp,
\end{eqnarray}
where 
\begin{eqnarray}
\pi^{m}_{1}(\phi)&=&\left( 1-\phi+\phi\dfrac{\rhonp}{\rhobf}\right)^{-1},\\
\pi^{m}_{2}(\phi)&=&\Pr\left( 1-\phi+\phi\dfrac{\rhonp}{\rhobf}\right)^{-1},\\
\pi^{m}_{3}(\phi)&=&\Pr \Ra_\text{nf} \M.
\end{eqnarray}
%~~~~~~~~~~~~~~~~~~~~~~~~~~~~~~~~
\section{Energy equation}
%~~~~~~~~~~~~~~~~~~~~~~~~~~~~~~~~
The dimensional energy equation writes
\begin{equation*}
\left(\udim\cdot\nabladim\Tempd\right) = \dfrac{1}{\cnf\rhonf} \nabladim\cdot\left(\knfdim\nabladim\Tempd\right) 
	+ \left(\dfrac{\rhonp\cnp}{\rhonf\cnf}\right) \left( \dfrac{D^{\star}_{\Temp}}{\TempdC} \nabladim\Tempd\cdot\nabladim\Tempd +D_{\omega}^{\star}\nabladim\phi\cdot\nabladim\Tempd   \right)
\end{equation*}
moving toward dimensionless variables the above equation writes
\begin{eqnarray*}
\dfrac{\alpha(\TempdH-\TempdC)}{L^2}\left(\u\cdot\nabla\Temp\right) &=& \dfrac{\kbf(\TempdH-\TempdC)}{L^2\cnf\rhonf} \nabla\cdot\left(\knf\nabla\Temp\right) \\
		&&+ \left(\dfrac{\rhonp\cnp}{\rhonf\cnf}\right)\dfrac{D_{\Temp_{c}}(\TempdH-\TempdC)^{2}}{\TempdC L^2} \left( D_{\Temp}\nabla\Temp\cdot\nabla\Temp  \right)\\
		&&+ \left(\dfrac{\rhonp\cnp}{\rhonf\cnf}\right)\dfrac{\phi_\text{b}D_{\omega_{c}}(\TempdH-\TempdC)}{L^2} \left(D_{\omega}\nabla\phi\cdot\nabla\Temp   \right).
\end{eqnarray*}
Multiplying the above equation by $\dfrac{L^2}{\alpha(\TempdH-\TempdC)}$ we obtain 
\begin{eqnarray*}
\left(\u\cdot\nabla\Temp\right) &=& \dfrac{\kbf}{\alpha\cnf\rhonf} \nabla\cdot\left(\knf\nabla\Temp\right) \\
		&&+ \left(\dfrac{\rhonp\cnp}{\rhonf\cnf}\right)\dfrac{D_{\Temp_{c}}(\TempdH-\TempdC)}{\alpha\TempdC} \left( D_{\Temp}\nabla\Temp\cdot\nabla\Temp  \right)\\
		&&+ \left(\dfrac{\rhonp\cnp}{\rhonf\cnf}\right)\dfrac{\phi_\text{b}D_{\omega_{c}}}{\alpha} \left(D_{\omega}\nabla\phi\cdot\nabla\Temp   \right),
\end{eqnarray*}

which rewrites using non-dimensional variables as follows

\begin{eqnarray*}
\left(\u\cdot\nabla\Temp\right) &=& \pi^{e}_{1} \nabla\cdot\left(\knf\nabla\Temp\right) 
						   +\pi^{e}_{2}  \left( D_{\Temp}\nabla\Temp\cdot\nabla\Temp  \right)
						   +\pi^{e}_{3} \left(D_{\omega}\nabla\phi\cdot\nabla\Temp   \right)
\end{eqnarray*}

Where 
\begin{equation*}\begin{array}{ccccc}
 \pi^{e}_{1} &=& \dfrac{\kbf}{\alpha\cnf\rhonf}&=&\left(\phi+(1-\phi)\dfrac{\rhobf\cbf}{\rhonp\cnp}\right)^{-1}\\
 \pi^{e}_{2} &=&  \left(\dfrac{\rhonp\cnp}{\rhonf\cnf}\right)\dfrac{\phi_\text{b}D_{\omega_{c}}}{\alpha}&=&\dfrac{\St\Pr}{\Sc}\dfrac{\TempdH-\TempdC}{\TempdH}\left(\phi+(1-\phi)\dfrac{\rhobf\cbf}{\rhonp\cnp}\right)^{-1}\\
 \pi^{e}_{3} &=&\left(\dfrac{\rhonp\cnp}{\rhonf\cnf}\right)\dfrac{D_{\Temp_{c}}(\TempdH-\TempdC)}{\alpha\TempdC}&=&\dfrac{\Pr}{\Sc}\phi_\text{b}\left(\phi+(1-\phi)\dfrac{\rhobf\cbf}{\rhonp\cnp}\right)^{-1}.
 \end{array}
\end{equation*}
%~~~~~~~~~~~~~~~~~~~~~~~~~~~~~~~~
\section{Nanoparticle transport equation}
%~~~~~~~~~~~~~~~~~~~~~~~~~~~~~~~~
The particle dimensional equation writes 
$$
\nabladim\cdot\nabladim\phidim = \nabladim\cdot\left( D_{\omega}^{\star} \nabladim\phidim+\dfrac{D_{\Tempd}^{\star}}{\TempdC}\nabladim\Tempd\right)
$$
using the non-dimensional equations the above equation writes 
$$
\dfrac{\alpha}{L^2}\phi_\text{b}\nabla\cdot\nabla\phi = \dfrac{\phi_\text{b}D_{\omega_{c}}}{L^2}\nabla\cdot\left( D_{\omega} \nabla\phi\right)
 										 +\dfrac{D_{\TempdC} (\TempdH-\TempdC)}{ L^2\TempdC} \nabla\cdot\left( D_{\Temp}\nabla\Temp\right).
$$

Multiplying the above by $\dfrac{L^2}{\alpha\phi_\text{b}}$ we obtain 

$$
\phi_\text{b}\nabla\cdot\nabla\phi = \dfrac{D_{\omega_{c}}}{\alpha}\nabla\cdot\left( D_{\omega} \nabla\phi\right)
 										 +\dfrac{D_{\TempdC}^{\star} (\TempdH-\TempdC)}{\phi_\text{b}\alpha\TempdC} \nabla\cdot\left( D_{\Temp}\nabla\Tempd\right).
$$

$$
\nabla\cdot\nabla\phi = \pi^{p}_{1}\nabla\cdot\left( D_{\omega} \nabla\phi\right)
				+\pi^{p}_{2} \nabla\cdot\left( D_{\Temp}\nabla\Tempd\right).
$$
where
\begin{equation*}
\begin{array}{ccccc}
\pi^{p}_{1} &=&  \dfrac{D_{\omega_{c}}}{\alpha}&=&\dfrac{1}{\Le},\vspace{.1in}\\
\pi^{p}_{2} &=& \dfrac{D_{\TempdC} (\TempdH-\TempdC)}{ L^2\TempdC}&=&\dfrac{\St\Pr}{\Sc}\dfrac{\TempdH-\TempdC}{\TempdC} \dfrac{1}{\phi_\text{b}}.
\end{array}
\end{equation*}

%~~~~~~~~~~~~~~~~~~~~~~~~~~~
\end{appendices}
%~~~~~~~~~~~~~~~~~~~~~~~~~~~

%% New version of the num-names style
\bibliographystyle{elsarticle-num-names}
\bibliography{refs.bib}
%~~~~~~~~~~~~~~~~~~~~~~~~~~~






% \end{thebibliography}

\end{document}